\documentclass[twocolumn,pra,showpacs,superscriptaddress]{revtex4-1}
\usepackage[latin9]{inputenc}
\usepackage{amsmath}
\usepackage{amssymb}
\usepackage{amsthm}
\usepackage{graphicx}
\usepackage{color}

\def\identity{\leavevmode\hbox{\small1\kern-3.8pt\normalsize1}}

\newtheorem{lemma}{Lemma}

\newcommand{\ket}[1]{\left | #1 \right\rangle}
\newcommand{\bra}[1]{\left \langle #1 \right |}

\newcommand{\Tr}{\mathrm{Tr}}

\newcommand{\avg}[1]{\left\langle #1\right\rangle}
\newcommand{\proj}[1]{\ket{#1}\bra{#1}}
\renewcommand{\epsilon}{\varepsilon}

\begin{document}

\title{Bell monogamy relations in arbitrary qubit networks}

\author{M.~C.~Tran}
\affiliation{School of Physical and Mathematical Sciences, Nanyang Technological University, 637371 Singapore}
\affiliation{Joint Center for Quantum Information and Computer Science, University of Maryland, College Park, Maryland 20742, USA}
\affiliation{Joint Quantum Institute, NIST/University of Maryland, College Park, Maryland 20742, USA}

\author{R.~Ramanathan}
\affiliation{Laboratoire d'Information Quantique, Universit\'e Libre de Bruxelles, Belgium}

\author{M.~McKague}
\affiliation{School of Electrical Engineering and Computer Science, Queensland University of Technology, Australia}

\author{D.~Kaszlikowski}
\affiliation{Centre for Quantum Technologies, National University of Singapore, 117543 Singapore, Singapore}
\affiliation{Department of Physics, National University of Singapore, 117543 Singapore, Singapore}

\author{T.~Paterek}
\affiliation{School of Physical and Mathematical Sciences, Nanyang Technological University, 637371 Singapore}
\affiliation{MajuLab, CNRS-UCA-SU-NUS-NTU International Joint Research Unit, UMI 3654 Singapore, Singapore}

\begin{abstract}
Characterizing trade-offs between simultaneous violations of multiple Bell inequalities in a large network has important physical consequences but is computationally demanding.
We propose a graph-theoretic approach to efficiently produce Bell monogamy relations in arbitrary arrangements of qubits.
All the relations obtained for bipartite Bell inequalities are tight and leverage only a single Bell monogamy relation.
This feature is unique to bipartite Bell inequalities, as we show that there is no finite set of such elementary monogamy relations for multipartite inequalities.
Nevertheless, many tight monogamy relations for multipartite inequalities can be obtained with our method as shown in explicit examples.
\end{abstract}

\maketitle

\section{introduction}
Bell monogamy relations describe the degree of simultaneous violation of multiple Bell inequalities.
They find applications in foundations of physics and quantum information.
On the fundamental side: 
(i) They were shown to exist in every no-signaling theory~\cite{PhysRevLett.87.117901,PhysRevA.65.012311,PhysRevA.71.022101,PhysRevA.73.012112,ProcRSocA.465.59,PhysRevLett.102.030403,PhysRevA.82.032313,PhysRevA.90.052323,PhysRevLett.113.210403}, but the principle of no-signaling alone does not single out the monogamies derived within the quantum framework~\cite{arXiv:0611001,PhysRevLett.106.180402,arXiv:1704.03790}.
They are therefore a natural test bed for candidate principles underlying quantum formalism~\footnote{Individual bipartite Bell inequalities already played an important role in this context, see e.g.~\cite{Pawlowski2009,ML}. 
However, intrinsically multi-partite principles are required~\cite{PhysRevLett.107.210403}, see e.g.~\cite{NatComms.4.2263}. 
Note that for testing these principles tight monogamy relations involving both bipartite and multipartite Bell inequalities are required.}.
(ii) They were also shown to play a role in quantum-to-classical transition~\cite{PhysRevLett.107.060405},
where tight Bell monogamy relations for multipartite inequalities reduce the number of particles for which the classical description emerges.
On the practical side: (i) They were used to obtain the upper bound on the average shrinking factor of cloning machines~\cite{PhysRevLett.102.030403}, which will here be refined~\footnote{Ref.~\cite{PhysRevLett.102.030403} derives the upper bound on the average shrinking factor from the principle of no-signaling. Our monogamy relations are derived within quantum formalism and put the upper bound on the average \emph{squared} shrinking factor. This is a stronger result, especially for anisotropic cloning machines~\cite{Cerf00}.}.
(ii) They were also shown to improve device-independent tasks such as randomness amplification and quantum key distribution~\cite{PhysRevA.82.032313,PhysRevA.90.052323}.
In the latter case, it is the existence of the Bell monogamy relation that allows for secure cryptography even in the presence of signaling~\cite{PhysRevA.82.032313}.
The method we propose here will contribute to generalizations of these results to multipartite cases.

Despite their importance, only a handful of Bell monogamy relations have been derived within the quantum formalism~\cite{PhysRevLett.87.117901,PhysRevA.65.012311,arXiv:0611001,PhysRevLett.106.180402,arXiv:1704.03790}.
A powerful approach to generate tight relations is given by the correlation complementarity~\cite{PhysRevLett.106.180402,PhysRevA.72.022340,JMathPhys.49.062105,NewJPhys.12.025009}.
The approach involves dividing relevant observables into sets of mutually anti-commuting ones.
The complexity of this task grows exponentially with the number of Bell parameters
and therefore renders correlation complementarity inefficient for large networks.
In fact, an efficient method to generate tight monogamy relations in arbitrary arrangements of a large number of qubits is not yet available. 
We propose such a method in this paper.

For every collection of bipartite Bell parameters, our method yields a corresponding tight monogamy relation.
The approach leverages a single Bell monogamy relation (derived in Ref.~\cite{arXiv:0611001}) multiple times.
We, therefore, name this monogamy relation as \emph{elementary}.
We then investigate if the method generalizes to multipartite inequalities.
It turns out that the situation is far more complicated already for tripartite inequalities.
We construct a Bell scenario with an increasing number of observers for which the method produces a Bell monogamy relation that is not tight,
even if all elementary relations for smaller numbers of observers are taken into account.
We conclude that already in tripartite scenario there is no finite set of elementary relations.
Nevertheless, the method does produce many tight monogamies and hence is valuable also in the multipartite case.


\section{Bell inequalities}
We focus on a complete set of correlation Bell inequalities for $m$ observers, each choosing between two measurement settings and obtaining a dichotomic $\pm 1$ outcome~\cite{PhysRevA.64.032112,PhysRevLett.88.210401}.
The set is equivalent to a single general Bell inequality $\mathcal{B}_{1\dots m} \le 1$ \footnote{See e.g. Eq. (5) in Ref.~\cite{PhysRevLett.88.210401}.}, where the following upper bound on the Bell parameter $\mathcal{B}_{1\dots m}$ was also derived:
\begin{equation}
	\mathcal{B}_{1\dots m}^2 \le \sum_{k_1 = x,y} \dots \sum_{k_m = x,y} T_{k_1 \dots k_m}^2\equiv \mathcal{T}_{1\dots m}^2.
	\label{T_BOUND_B}
\end{equation}
The summation is over orthogonal local directions $x$ and $y$ which span the plane of local settings and $T_{k_1 \dots k_m} = \Tr\left(\rho.\sigma_{k_1}\otimes\dots\otimes\sigma_{k_m}\right)$ are the quantum correlation functions of state $\rho$.
Therefore, if $\mathcal{T}_{1\dots m}^2 \le 1$ the quantum correlations admit a local hidden variable model (for measurements in the $xy$ plane).
The condition is also necessary and sufficient if $m=2$~\cite{PhysLettA.200.340}.
We shall use it as a building block for our monogamy relations.

\section{Bipartite inequalities}
Let us first consider trade-offs between simultaneous violations of a set of bipartite Bell inequalities.
Each bipartite Bell parameter may involve two out of $n$ observers, each having access to a single qubit. 
The problem can be represented by a graph with $n$ vertices denoting the $n$ observers and edges denoting the relevant Bell parameters.
An example of such a graph is given in Fig.~\ref{FIG_BIPARTITE}.
The simplest scenario of Bell monogamy is when three observers try to simultaneously violate two Bell inequalities.
The statement of Bell monogamy is that the simultaneous violation is impossible
and the quantitative quantum relation reads \cite{arXiv:0611001}:
\begin{equation}
\mathcal{B}_{12}^2 + \mathcal{B}_{13}^2 \le 2.
\label{3MON}
\end{equation}
This monogamy relation is a straightforward application of the correlation complementarity (see Appendix~\ref{APP_CC}).

\begin{figure}[t]
\includegraphics[width=0.45\textwidth]{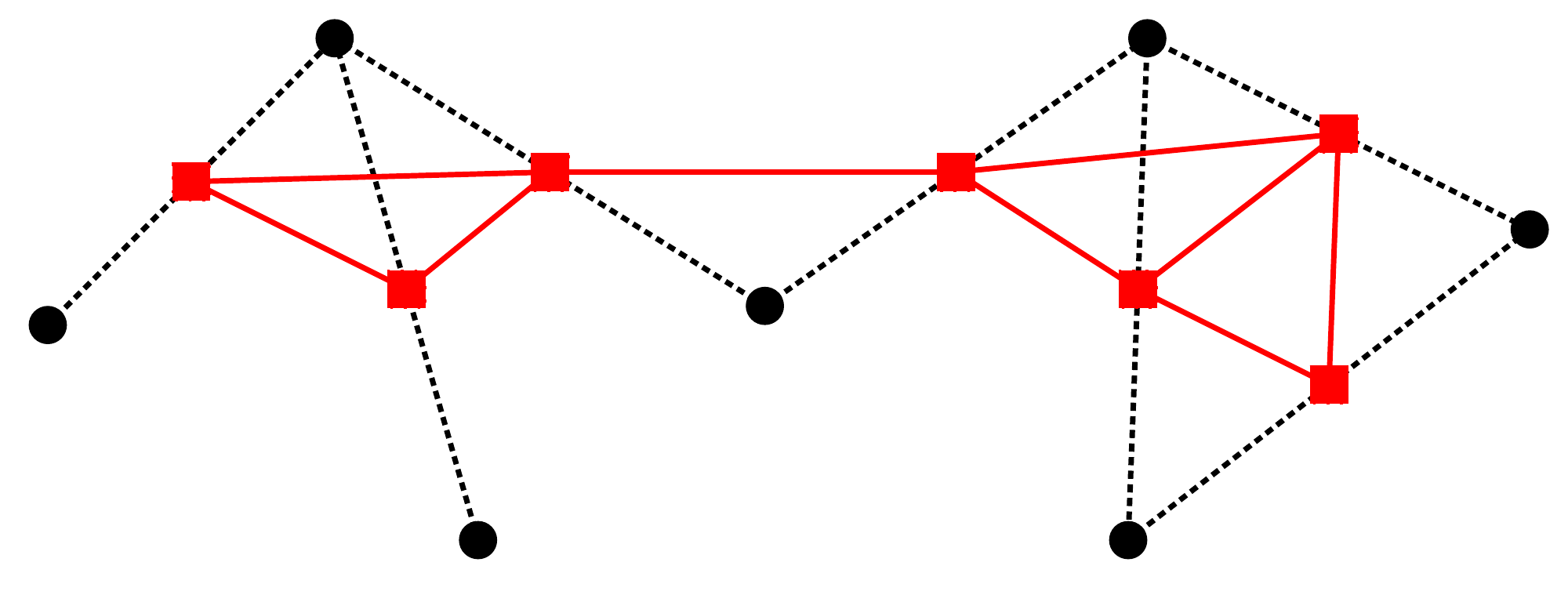}
\caption{Tight Bell monogamy of bipartite Bell inequalities.
The vertices (circles) of the black dotted graph $G$ represent different observers and edges connect observers who test whether a bipartite Bell inequality is violated.
The red solid graph, with squares as vertices, is a line graph of $G$. 
Its vertices represent Bell inequalities and edges connect Bell inequalities which share a common observer.
Properties of this construct determine tight Bell monogamy relations (\ref{BI_MONO}).}
\label{FIG_BIPARTITE}
\end{figure}

For a general graph, one can in principle also apply correlation complementarity to find tight monogamy relations.
The method requires grouping relevant observables into mutually anti-commuting sets, which is computationally demanding.
Instead, we propose the following simple method to derive a tight Bell monogamy relation for every graph.
Denote by $G$ the graph with observers represented by vertices and Bell inequalities by edges.
A line graph $L$ of the initial graph $G$ is constructed by placing vertices of $L$ on every edge of $G$,
and by connecting the vertices of $L$ whenever the corresponding edges of $G$ share a vertex (Fig.~\ref{FIG_BIPARTITE}).
The properties of the line graph determine the Bell monogamy relations.
Note that Bell inequalities are represented by vertices of $L$ and edges of $L$ provide information whether two Bell inequalities share a common observer.
In other words, for every edge of $L$ we have monogamy relation (\ref{3MON}), and summing them up gives a general monogamy
\begin{equation}
\sum_{v \in L} d_v \mathcal{B}_v^2 \le 2 \epsilon,
\label{BI_MONO}
\end{equation}
where the sum is over the vertices of $L$, $d_v$ denotes the number of edges incident to the vertex $v$,
$\mathcal{B}_v$ is the Bell parameter associated with vertex $v$ and $\epsilon$ is the total number of edges in $L$.
The factor of $2$ comes from the monogamy relation (\ref{3MON}).
We shall refer to this method as the \emph{averaging} method.

The general monogamy relation (\ref{BI_MONO}) turns out to be tight, i.e. the bound cannot be any smaller.
This follows from the handshaking lemma that for any finite undirected graph, $\sum_v d_v = 2 \epsilon$.
This corresponds to $\mathcal{B}_v = 1$ for all the vertices of $L$, achieved e.g. by the state $\ket{\uparrow \dots \uparrow}$, where all the spins are aligned along the $x$ axis and the measurements are all $\sigma_x$.

We emphasize that this construction is general and surprisingly simple.
It applies to arbitrary graphs, i.e. an arbitrary number of observers measuring an arbitrary configuration of the bipartite Bell inequalities while in the process only monogamy relation (\ref{3MON}) is utilized.
We therefore term the monogamy relation (\ref{3MON}) \emph{elementary}.

On a side note, the elementary relation (\ref{3MON}) has a remarkable property that all mathematically allowed values of $\mathcal{B}_{12}$ and $\mathcal{B}_{13}$ that saturate it are physically realizable~\cite{arXiv:0611001}.
The general monogamy relation (\ref{BI_MONO}) does not share this property as simply seen by considering the triangle graph: 
in this case (\ref{BI_MONO}) gives the bound of $3$, while each individual Bell expression can take at most the maximum Tsirelson value of $\sqrt{2}$.
However, one may ask if the set defined by the \emph{intersection} of elementary relations (\ref{3MON}) contains values of Bell parameters that are all physically achievable.
We show in 
the Appendix~\ref{APP_SEC_SOLID} 
examples of configurations where all the points in the intersection are indeed realized in quantum physics.

A natural question is whether the averaging method generalizes to multipartite Bell monogamy, i.e. $m>2$.
In particular, we ask if there exists an elementary monogamy relation, or a finite set of elementary monogamy relations, from which tight monogamy relations could be derived in an arbitrary scenario.
The answer is more complex even for tripartite Bell inequalities.
On one hand, there are simple monogamy relations averaging which results in tight monogamy relations.
But on the other hand, there are Bell scenarios where tight monogamy relations cannot be obtained from simpler relations.
We now discuss them in more detail.


\section{Tripartite inequalities}
The graph-theoretic approach from the previous section can be naturally extended to tripartite Bell inequalities.
The graph $G$ is now upgraded to a hypergraph with the vertices representing observers and hyperedges connecting three observers testing a violation of the Bell inequality.
Fig.~\ref{FIG_MONO_42} presents examples of such hypergraphs.
The two Bell monogamy relations at the bottom of Fig.~\ref{FIG_MONO_42} list all possible ways two tripartite Bell parameters may overlap and results from the bipartite case suggest they might be of special importance. 
Using correlation complementarity, one easily verifies that the corresponding monogamy relations hold (see Appendix~\ref{APP_CC}):
\begin{eqnarray}
\mathcal{B}_{123}^2 + \mathcal{B}_{124}^2 & \leq & 4, \label{EQ_2_OVER} \\
\mathcal{B}_{123}^2 + \mathcal{B}_{145}^2 & \leq & 4. \label{EQ_3_OVER}
\end{eqnarray}
The averaging method naturally extends. But now in the line graph, an edge connects Bell inequalities that share at least one common observer.
Since the bound in both inequalities above is the same, the general monogamy relation is of the form (\ref{BI_MONO}) with the factor of $2$ on the right-hand side replaced by $4$.

\begin{figure}[t]
\includegraphics[width=0.35\textwidth]{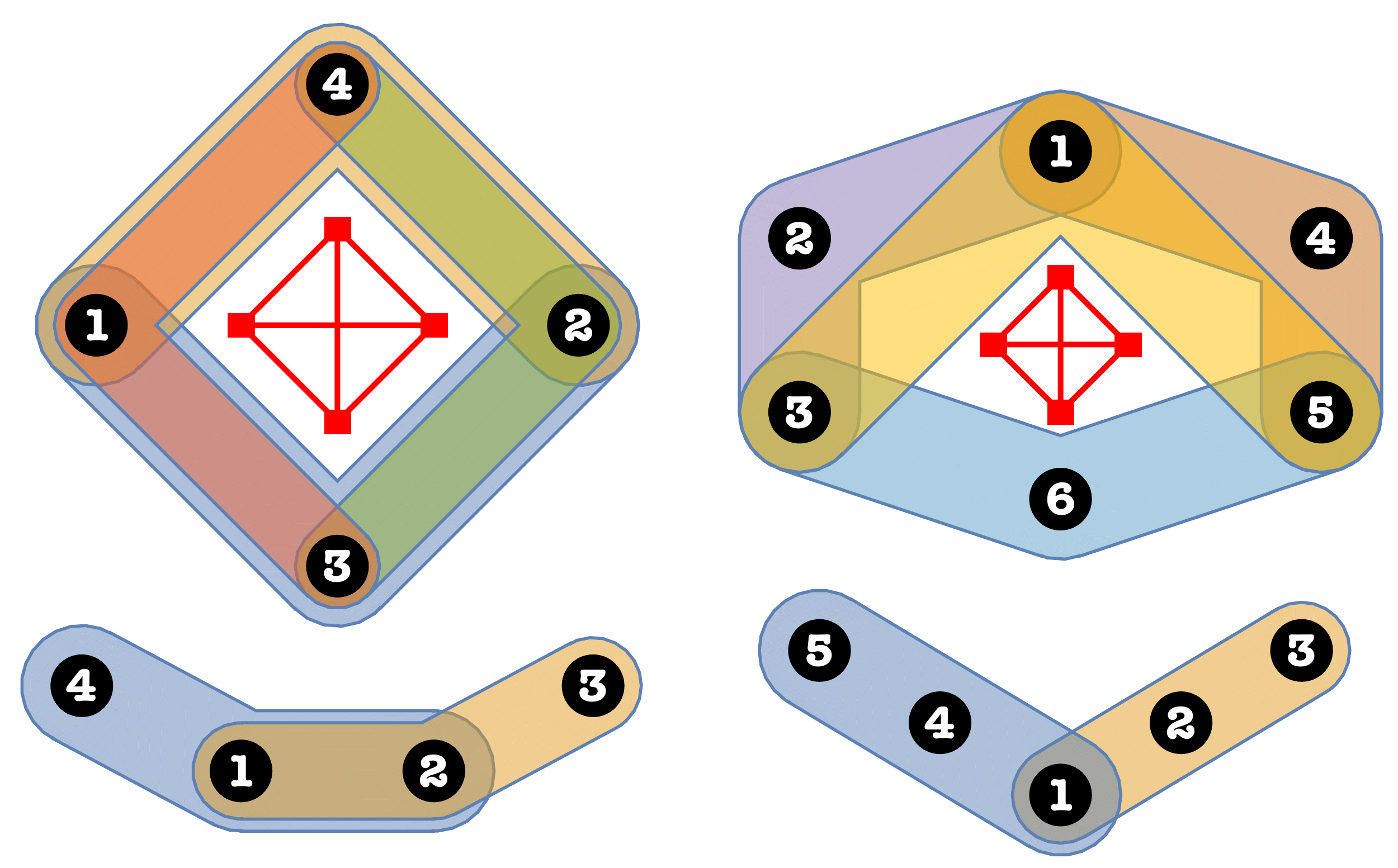}
\caption{Examples of hypergraphs describing Bell monogamy relations of tripartite Bell inequalities. The vertices correspond to qubits and the hyperedges, represented by different shaded areas (color online), are Bell inequalities.}
\label{FIG_MONO_42}
\end{figure}

Note, however, that the physical implication of monogamy relations (\ref{EQ_2_OVER}) and (\ref{EQ_3_OVER}) is different than that of the bipartite relation (\ref{3MON}).
In the bipartite case, whenever one Bell inequality is violated, the other has to be satisfied.
However, two tripartite Bell inequalities can be violated simultaneously.
There indeed exist quantum states and measurements which give rise to any value of the Bell parameters compatible with (\ref{EQ_2_OVER}) and (\ref{EQ_3_OVER})~\cite{PhysRevLett.106.180402,arXiv:1704.03790}.
This suggests that there are Bell monogamy relations stronger than the two listed above.

A concrete such Bell monogamy relation is presented at the top left corner of Fig.~\ref{FIG_MONO_42}.
It is a very condensed graph where four observers aim at testing four tripartite Bell inequalities.
The red graph is the line graph of the original hypergraph.
Since each $d_v = 3$ and there are $\epsilon = 6$ edges of $L$ in total, the averaging method predicts
$\mathcal{B}_{123}^2 + \mathcal{B}_{234}^2 + \mathcal{B}_{341}^2 + \mathcal{B}_{412}^2 \le 8$, whereas the tight bound is~\cite{PhysRevLett.106.180402}:
\begin{equation}
\mathcal{B}_{123}^2 + \mathcal{B}_{234}^2 + \mathcal{B}_{341}^2 + \mathcal{B}_{412}^2 \le 4.
\label{EQ_SQUARE}
\end{equation}
Accordingly, the monogamy relations (\ref{EQ_2_OVER}) and (\ref{EQ_3_OVER}) do not form a set of elementary relations from which tight monogamy relations can be derived using averaging method for all more complicated hypergraphs.

Since the monogamy relation (\ref{EQ_SQUARE}) involves four Bell parameters and it is bounded by $4$, it shares with relation (\ref{3MON}) its physical implication.
Namely, if one Bell inequality is violated, another must be satisfied.
It is therefore interesting to augment the set of monogamy relations $\{ (\ref{EQ_2_OVER}), (\ref{EQ_3_OVER})\}$ with inequality (\ref{EQ_SQUARE}) and verify which tight monogamy relations follow from the averaging method.
In fact, since (\ref{EQ_2_OVER}) is a special case of (\ref{EQ_SQUARE}), it is sufficient to replace one with the other.
Likewise, (\ref{EQ_3_OVER}) is a special case of the following monogamy relation, presented at the top right corner of Fig.~\ref{FIG_MONO_42},
\begin{align}
\mathcal{B}_{123}^2 + \mathcal{B}_{145}^2 + \mathcal{B}_{135}^2 + \mathcal{B}_{356}^2 \le 4.
\label{EQ_THE_AVENGERS}
\end{align}
This leads the question whether a finite set of elementary monogamy relations exists, i.e. such a set that the averaging method produces tight monogamy relation for arbitrary hypergraphs.
Note that when adding \eqref{EQ_SQUARE} and \eqref{EQ_THE_AVENGERS} to the set of elementary relations, the line graph method must be suitably updated.
Since the relations \eqref{EQ_SQUARE} and \eqref{EQ_THE_AVENGERS} involve more than two Bell parameters, an edge in the line graph may connect more than two vertices. 
Therefore the line graph needs to be upgraded to a hypergraph.
In contrast to the bipartite case, there may be more than one line hypergraph for each original hypergraph. 
We shall take into account all possible line hypergraphs in the averaging method.

We may attempt to construct the set of elementary relations by a brute force algorithm searching over all hypergraphs with $n$ vertices and $h$ hyperedges, each covering $m$ vertices (see Appendix~\ref{APP_ALG}).
In principle, if this algorithm were to be run for infinitely large $n$, it returns a set of elementary monogamy relations $\mathcal{E}$.
We shall now argue that such a set must, in fact, be infinite, in stark contrast to the case of bipartite Bell inequalities.


\section{The infinite set}
We shall construct a set of hypergraphs with increasing number of vertices and show that their corresponding monogamy relations obtained using the averaging method are not tight,
even if the set $\mathcal{E}$ is composed of all elementary monogamy relations with smaller numbers of observers.
A part of the set is depicted in Fig.~\ref{FIG_MONO_C357}.
We consider cyclic hypergraphs $C_h$ that involve an odd number of Bell inequalities, $h = 3,5,\dots$, which are tested by $2 h$ observers. 

\begin{figure}[b]
\includegraphics[width=0.45	\textwidth]{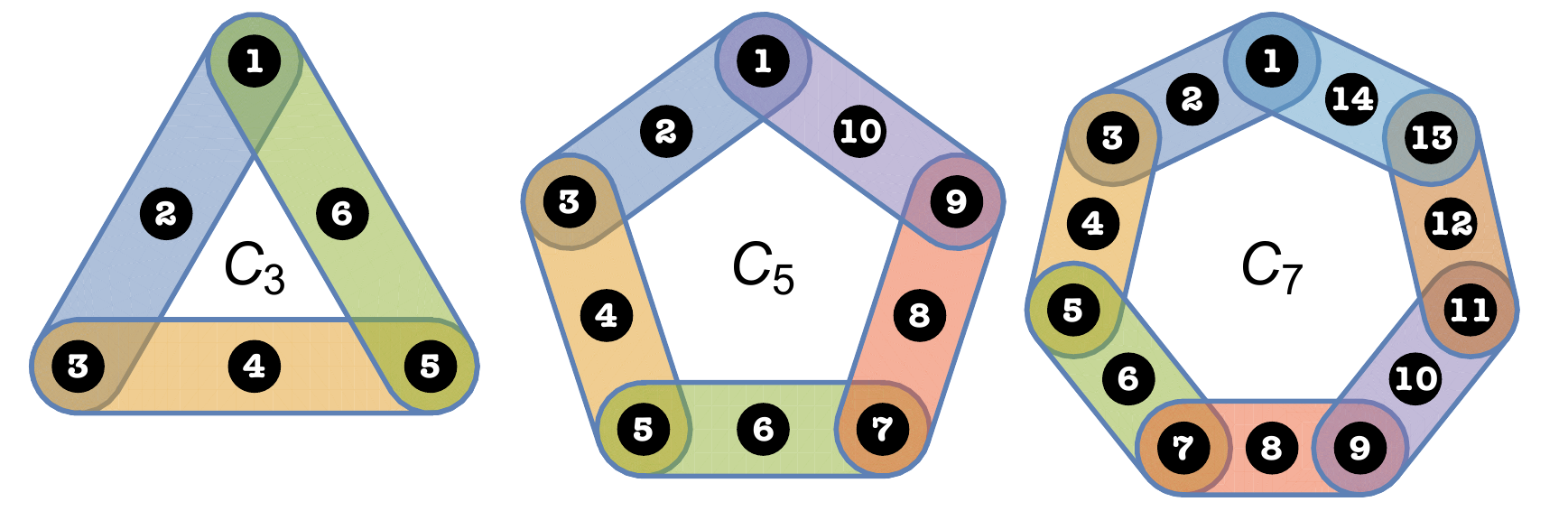}
\caption{First three cyclic hypergraphs we use to show that there is no finite set of elementary monogamy relations of tripartite Bell inequalities.}
\label{FIG_MONO_C357}
\end{figure}

Assume that the algorithm above returned a finite set $\mathcal{E}$ and let $n$ be the highest number of observers involved in any elementary monogamy relation in $\mathcal{E}$.
We begin our analysis with the graph $C_h$ that has the number of vertices higher than $n$.
In this way, we rule out the case that $C_h$ is (a subgraph of a hypergraph) already present in $\mathcal{E}$.
Therefore, the only way of obtaining a bound on the monogamy relation corresponding to $C_h$ is to combine graphs or subgraphs of monogamy relations in $\mathcal{E}$ that simultaneously are the subgraphs of $C_h$.
The only non-trivial subgraphs of $C_h$ are connected graphs involving $b$ consecutive Bell parameters.
The case of $b = 2$ is covered by the monogamy relation (\ref{EQ_3_OVER}), for which the bound is achieved, e.g. if the first three particles are in the Greenberger-Horne-Zeilinger (GHZ) state~\cite{GHZ}.
For any higher $b$, the corresponding monogamy relation has to have the bound of at least $2b$ as this is the number obtained if the triples of particles tested in every second Bell parameter are in the GHZ state.
Since our method is averaging these monogamy relations, it follows that the Bell monogamy relation corresponding to $C_h$ has the bound of at least $2h$.
A concrete example how the bound of $2 h$ is obtained is presented in Fig.~\ref{FIG_CONSEC}.

\begin{figure}[t]
\includegraphics[width=0.45\textwidth]{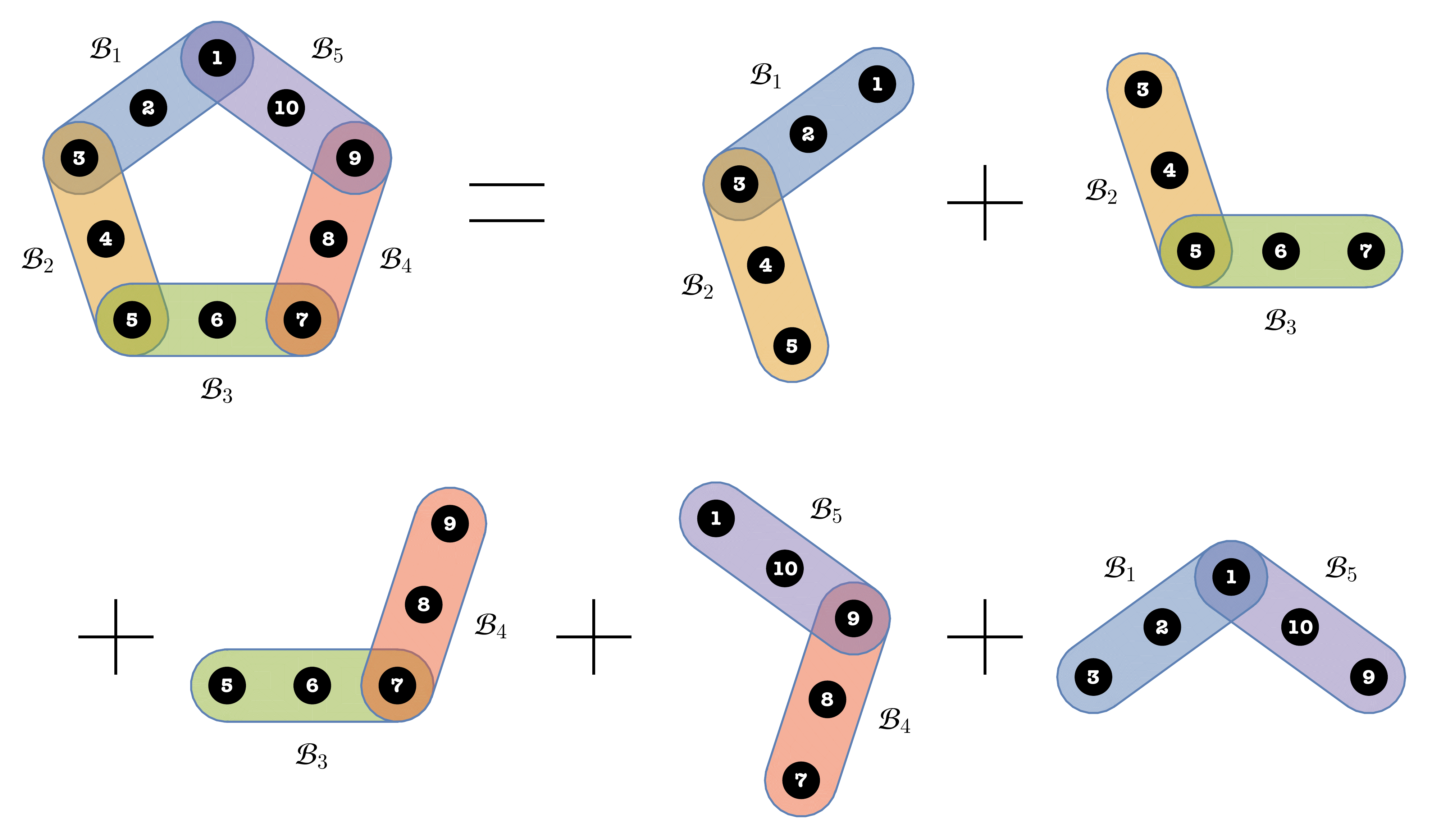}
\caption{Every elementary Bell monogamy relation to the right involves two consecutive Bell parameters and has the bound of $4$, see (\ref{EQ_3_OVER}). 
The line graph corresponding to the cyclic hypergraph on the left is a pentagon and leads to the monogamy relation
$\mathcal{B}_{1}^2 + \mathcal{B}_{2}^2 + \mathcal{B}_{3}^2 + \mathcal{B}_{4}^2 + \mathcal{B}_{5}^2 \le 10$.}
\label{FIG_CONSEC}
\end{figure}

We now show that the bound $2h$ is not tight, i.e. there is no quantum state and measurements achieving it (recall that $h$ is odd).
We point out properties that such a hypothetical state would have to satisfy and show that they are contradictory.
Let us label the Bell parameters in $C_h$ by index $j = 1, \dots,h$.
A way to obtain the bound $2h$ is presented in Fig.~\ref{FIG_CONSEC} and involves summing up pairs of consecutive Bell parameters $\mathcal{B}_{j}^2 + \mathcal{B}_{j+1}^2 \leq 4$, for $j = 1, \dots, h$, with $h+1 \equiv 1$.
Therefore, saturation of the bound of $2h$ implies saturation of every constituent monogamy, i.e. $\mathcal{B}_{j}^2 + \mathcal{B}_{j+1}^2 = 4$ for all $j$'s in question.
Recall that the bound $\mathcal{B}_{j}^2 + \mathcal{B}_{j+1}^2 \leq 4$ is proved by partitioning $16$ observables that enter the upper bound (\ref{T_BOUND_B}) into 4 groups, each of $4$ mutually anti-commuting observables (see Appendix~\ref{APP_CC}).
According to correlation complementarity, the constituent monogamy is saturated if each group of anti-commuting observables saturates the bound of $1$.
In particular, we have $\mathcal{X}_{j} + \mathcal{X}_{j+1} = 1$, where $\mathcal{X}_{j}$ is defined as
\begin{equation}
\mathcal{X}_{j} \equiv \avg{X_{2j-1} X_{2j} Y_{2j+1}}^2 + \avg{X_{2j-1}Y_{2j}Y_{2j+1}}^2,
\end{equation}
where e.g. $X_{2j}$ denotes Pauli-$x$ operator acting on the qubit $2j$.
Since there is an odd number of Bell parameters, each $\mathcal{X}_{j}$ must be exactly $1/2$.
It is perhaps worth a comment that one cannot proceed any further using correlation complementarity alone.
We shall now utilize the relation between correlations and marginal expectation values 
introduced previously in the context of non-local hidden variable theories~\cite{FoundPhys.33.1469,Nature.446.871,NaturePhys.4.681}.
We construct an observable that on one hand necessarily has high expectation value but on the other hand it must have small average 
as it anti-commutes with observables that enter $\mathcal{X}_{1}$ and $\mathcal{X}_{2}$, leading to a final contradiction.

Let us consider $\mathcal{X}_{3}$.
We introduce observable $M_6 = \alpha \, X_5 X_6 Y_7 + \beta \, X_5 Y_6 Y_7$, with normalized vector $(\alpha,\beta)$ parallel to $(\avg{X_{5} X_{6} Y_{7}},\avg{X_{5} Y_{6} Y_{7}})$.
It has expectation value $\avg{M_6} = 1/\sqrt{2}$, because $\mathcal{X}_{3} = 1/2$.
Similarly, we find observable $M_{2h}$ with expectation value $\avg{M_{2h}} = 1/\sqrt{2}$ following from $\mathcal{X}_{h} = 1/2$.
For $h \ge 5$ the two observables $M_6$ and $M_{2h}$ have no overlapping qubits and can be measured simultaneously. 
Their product satisfies the lower bound (see Appendix~\ref{APP_LEGGETT}):
\begin{equation}
\avg{M_6 M_{2h}}^2 \ge \left(|\avg{M_6}| + |\avg{M_{2h}}| - 1 \right)^2 = (\sqrt2-1)^2.
\label{EQ_MM_BOUND}
\end{equation}
At the same time one verifies that observable $M_6 M_{2h}$ together with observables entering $\mathcal{X}_{1}$ and $\mathcal{X}_{2}$ form a pairwise anti-commuting set. 
Therefore, by the correlation complementarity,
\begin{equation}
\mathcal{X}_{1} + \mathcal{X}_{2} + \avg{M_6 M_{2h}}^2 \le 1.
\label{EQ_XX_CC}
\end{equation}
Inequalities \eqref{EQ_MM_BOUND} and \eqref{EQ_XX_CC} contradict $\mathcal{X}_{1} = \mathcal{X}_{2} = 1/2$.
Summing up, there is no state and measurements which achieve the bound of $2h$ derived from the averaging method applied on a sequence of cyclic hypergraphs $C_h$.
Accordingly, the set of elementary tripartite monogamy relations $\mathcal{E}$ must contain an infinite number of monogamy relations --- at least those that correspond to all $C_h$'s with odd $h$.

\begin{figure}[t]
	\includegraphics[width=0.5\textwidth]{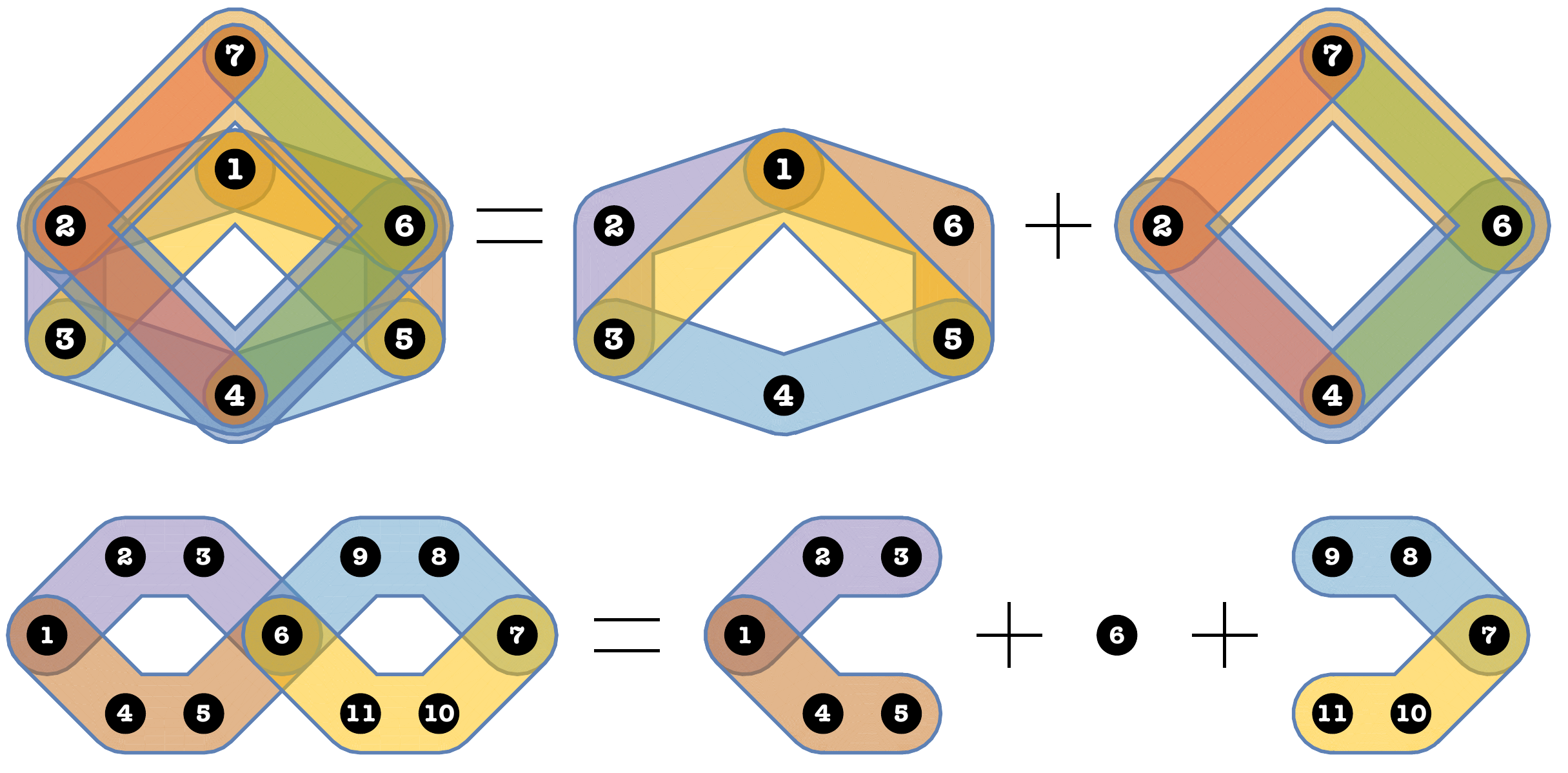}
	\caption{Different ways tight monogamy relations can be obtained from elementary monogamy relations. 
		In the upper row, a tight monogamy relation for the graph to the left, i.e. $\mathcal{B}_{123}^2+\mathcal{B}_{345}^2+\mathcal{B}_{561}^2+\mathcal{B}_{135}^2+\mathcal{B}_{246}^2+\mathcal{B}_{467}^2+\mathcal{B}_{672}^2+\mathcal{B}_{724}^2\leq 8$, is obtained from averaging two different elementary monogamy relations, namely \eqref{EQ_SQUARE} and \eqref{EQ_THE_AVENGERS}. 
		In the lower row, a tight monogamy relation of four-partite Bell inequalities is obtained from the elementary monogamy relations of tripartite Bell inequalities in Eq.~\eqref{EQ_3_OVER}. }
	\label{FIG_combine_2_diff}
\end{figure}

\section{Optimistic coda}
We showed a simple method to produce Bell monogamy relations for an arbitrary arrangement of observers and Bell inequalities.
For bipartite inequalities, the monogamy relations obtained are tight and leverage a single elementary monogamy relation derived in Ref.~\cite{arXiv:0611001}.
Using a combination of correlation complementarity and inequalities for marginal expectation values, we showed that in the multipartite case, however, there is no finite set of elementary monogamy relations from which tight relations corresponding to arbitrarily complicated graphs could be obtained.
Nevertheless, the method is useful as it does produce non-trivial tight Bell monogamy relations.
For example, a tight monogamy relation for a completely connected graph is obtained by averaging (\ref{EQ_SQUARE}) only,
combining elementary monogamy relations of different types can give tight relations, e.g. top of Fig.~\ref{FIG_combine_2_diff}, 
as well as tight relations for a higher number of observers, e.g. at the bottom of Fig.~\ref{FIG_combine_2_diff} and in the Appendix~\ref{APP_N}.
We hope our general method will boost further applications of Bell monogamy, especially in complex multiparty scenarios.

\begin{acknowledgments}
We thank Pawe{\l} Kurzy\'nski for discussions.
This work is supported by the Singapore Ministry of Education Academic Research Fund Tier 2 project MOE2015-T2-2-034.
M.C.T. acknowledges support from the NSF-funded Physics Frontier Center at the JQI and the QuICS Lanczos Graduate Fellowship.
R.R. acknowledges support from the research project ``Causality in quantum theory: foundations and applications'' of the Fondation Wiener-Anspach 
and from the Interuniversity Attraction Poles 5 program of the Belgian Science Policy Office under the grant IAP P7-35 photonics@be.
This work is supported by Singapore Ministry of Education Academic Research Fund Tier 3 (Grant No. MOE2012-T3-1-009).
\end{acknowledgments}

\appendix


\section{Correlation complementarity and Bell monogamy relations}
\label{APP_CC}

Here we give examples how correlation complementarity yields Bell monogamy relations. 
In particular, we detail groups of anti-commuting observables that lead to monogamy relations \eqref{3MON} and \eqref{EQ_3_OVER}, 
the latter used in the proof that the set of elementary tripartite relations is infinite.

Given a set of dichotomic, mutually anti-commuting observables 
$\left\{O_j\right\}$, correlation complementarity states that $\sum_j\avg{O_j}^2\leq 1$, where $\avg{O_j}=\Tr(\rho.O_j)$ are the expectation values in the state $\rho$~\cite{PhysRevLett.106.180402,PhysRevA.72.022340,JMathPhys.49.062105,NewJPhys.12.025009}.
In the first example, let us show how to use correlation complementarity to prove Eq.~\eqref{3MON}:
\begin{align}
\mathcal{B}_{12}^2+\mathcal{B}_{13}^2\leq 2.
\end{align}
The relevant observables that enter into the upper bound \eqref{T_BOUND_B} are $X_1X_2, X_1Y_2, Y_1X_2, Y_1Y_2$ (for $\mathcal{B}_{12}$) and $X_1X_3, X_1Y_3, Y_1X_3, Y_1Y_3$ (for $\mathcal{B}_{13}$).
These eight observables can be partitioned into two sets, namely $\left\{X_1X_2,X_1Y_2,Y_1X_3,Y_1Y_3\right\}$ and $\left\{X_1X_3,X_1Y_3,Y_1X_2,Y_1Y_2\right\}$, each containing only mutually anti-commuting observables. 
Eq.~\eqref{3MON} follows by direct application of correlation complementarity to these two sets.\\

Similarly, $\mathcal{B}_{123}^2+\mathcal{B}_{145}^2$ is upper bounded by the sum of squared expectation values of $16$ observables.
They can be arranged into $4$ groups of mutually anti-commuting observables, e.g. the $4$ columns of the following table:
\begin{equation}
\label{TAB_CC}
\centering
\begin{tabular}{|c|c|c|c|}
\hline
$X_1X_2X_3$ & $Y_1X_2Y_3$ & $X_1X_2Y_3$ & $Y_1X_2X_3$ \\ 
$X_1Y_2X_3$ & $Y_1Y_2Y_3$ & $X_1Y_2Y_3$ & $Y_1Y_2X_3$ \\ 
$Y_1X_4Y_5$ & $X_1X_4X_5$ & $Y_1X_4X_5$ & $X_1X_4Y_5$ \\  
$Y_1Y_4Y_5$ & $X_1Y_4X_5$ & $Y_1Y_4X_5$ & $X_1Y_4Y_5$ \\ 
\hline
\end{tabular}
\end{equation}
Correlation complementarity gives a bound of 1 for each group, and hence a bound of 4 for $\mathcal{B}_{123}^2+\mathcal{B}_{145}^2$.


\section{Configurations with quantum trade-off completely characterized by Steinmetz solids}
\label{APP_SEC_SOLID}

We shall give two examples of configurations where the intersection of the elementary monogamy relations precisely captures the trade-off relation within quantum theory.
The first example is the star configuration with arbitrary number of observers, see Fig.~\ref{APP_FIG_STAR}, and the second example involves four observers in the chain configuration, see Fig.~\ref{FIG_APP_CHAIN}.
The resulting quantum sets of allowed values of Bell parameters are intersections of cylinders having the same radii, the sets known as the Steinmetz solids.
We note that Toner and Verstraete already realized that the Steinmetz solid corresponding to the triangle configuration is \emph{not} the quantum set, it contains points which cannot be realized in quantum physics~\cite{arXiv:0611001}.

\subsection{Star configuration}
\label{APP_SEC_STAR}

\begin{figure}[b]
	\includegraphics[width=0.45\textwidth]{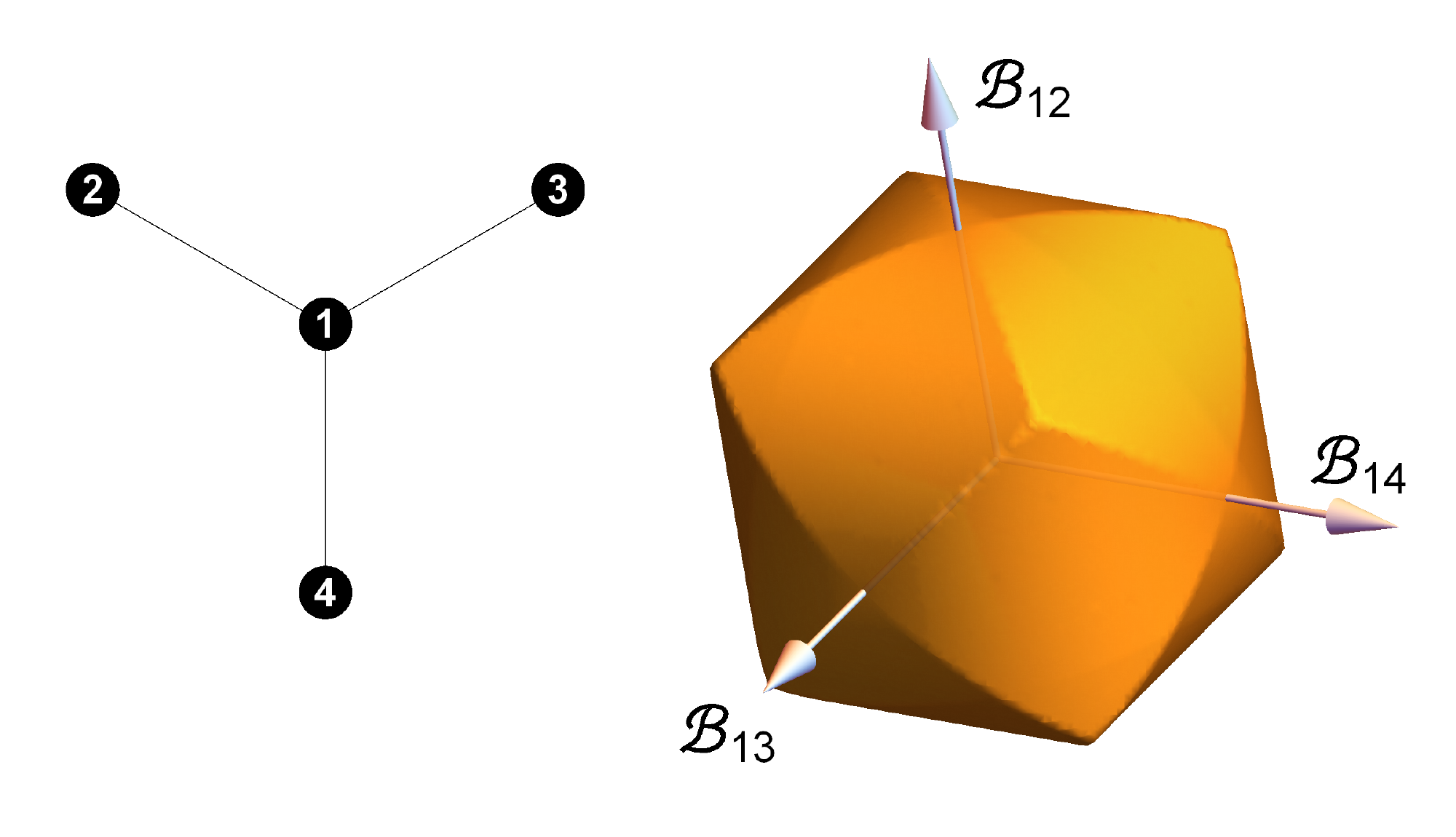}
	\caption{A star configuration of four observers (left) and the corresponding Steinmetz solid (right).
		All the points within the solid, and only those points, are achievable within quantum theory.
		Similar solids fully characterize the quantum set of Bell parameters for arbitrary number of satellite parties with common first observer.}
	\label{APP_FIG_STAR}
\end{figure}

Consider $n+1$ observers arranged in a star network with a central Alice and $n$ remaining observers, Fig.~\ref{APP_FIG_STAR}.
We show that any set of values $\mathcal{B}_{1i}$ that obeys $\mathcal{B}^2_{1j} + \mathcal{B}^2_{1k} \leq 2$ for $i, j, k \in \{2, \dots, n+1\}$ is achievable by a suitable shared quantum state and measurements.

Let us take a Cartesian coordinate system in $n$ dimensions with each axis giving the value of $\mathcal{B}_{1i}$.
Clearly, any point in the intersection of the cylinders $\mathcal{B}^2_{1j} + \mathcal{B}^2_{1k} \leq 2$ can have only one coordinate, say $\mathcal{B}_{12}$, larger than the local bound of $1$. 
At this point, the remaining Bell expressions can attain at most the value $\sqrt{2 - \mathcal{B}^2_{12}}$. The shared quantum state achieving these values is given as
\begin{eqnarray}
\label{eq:opt-state}
| \psi \rangle = (\alpha |00 \rangle + \beta |11 \rangle)_{12} | 0 \rangle_3 \dots | 0 \rangle_{n+1}
\end{eqnarray}   
with $\alpha = \frac{1}{\sqrt{2}} \sqrt{1+\sqrt{2} \sin t}$ and $\beta = \frac{1}{\sqrt{2}} \sqrt{1 - \sqrt{2} \sin t}$, with a parameter $0 \leq t \leq \frac{\pi}{4}$.
The relevant correlations are in the $xz$ plane and give
\begin{eqnarray}
&&\mathcal{B}^2_{12} = \sum_{k_1=x,z} \sum_{k_{2} = x,z} T^2_{k_1,k_2} = 2 \cos^2{t}, \nonumber \\
&&\mathcal{B}^2_{1j} = \sum_{k_1=x,z} \sum_{k_{j} = x,z} T^2_{k_1,k_j} = 2 \sin^2{t}, \quad \forall j > 2.
\end{eqnarray}
For the range $0 \leq t \leq \frac{\pi}{4}$, we see that $\mathcal{B}_{12}$ violates the local bound of $1$, and every relation of the form $\mathcal{B}^2_{12} + \mathcal{B}^2_{1j} \leq 2$ is saturated. 

Moreover, there is a freedom to control the parameters $\mathcal{B}_{1j}$ for $j > 2$ such that not all of them achieve the maximum possible value of $2 \sin^2{t}$ in this situation. 
To reduce the value of any individual $\mathcal{B}_{1j}$, we replace in (\ref{eq:opt-state}) the state $|0 \rangle_j$ at position $j$ by the noisy state $p_j |0 \rangle_j \langle 0| + \frac{1-p_j}{2} \mathbb{I}$, where $\mathbb{I}$ denotes the $2 \times 2$ identity matrix.
The new state then gives
\begin{eqnarray}
&&\mathcal{B}^2_{12} = \sum_{k_1=x,z} \sum_{k_{2} = x,z} T^2_{k_1,k_2} = 2 \cos^2{t}, \nonumber \\
&&\mathcal{B}^2_{1k} = \sum_{k_1=x,z} \sum_{k_{j} = x,z} T^2_{k_1,k_j} = 2 \sin^2{t}, \quad \forall k > 2, j \neq k \nonumber \\
&&\mathcal{B}^2_{1j} = \sum_{k_1=x,z} \sum_{k_{2} = x,z} T^2_{k_1,k_2} = 2 p_j^2 \sin^2{t}.
\end{eqnarray}
By controlling the noise levels $p_j$ we see that the strategy allows to achieve every possible point within the Steinmetz solid.
The intersection of the cylinders is therefore precisely the shape of the set of $n+1$-party quantum correlations projected onto the space of two-party (CHSH) Bell parameters.   

\subsection{Chain configuration}

\begin{figure}[]
	\includegraphics[width=0.45\textwidth]{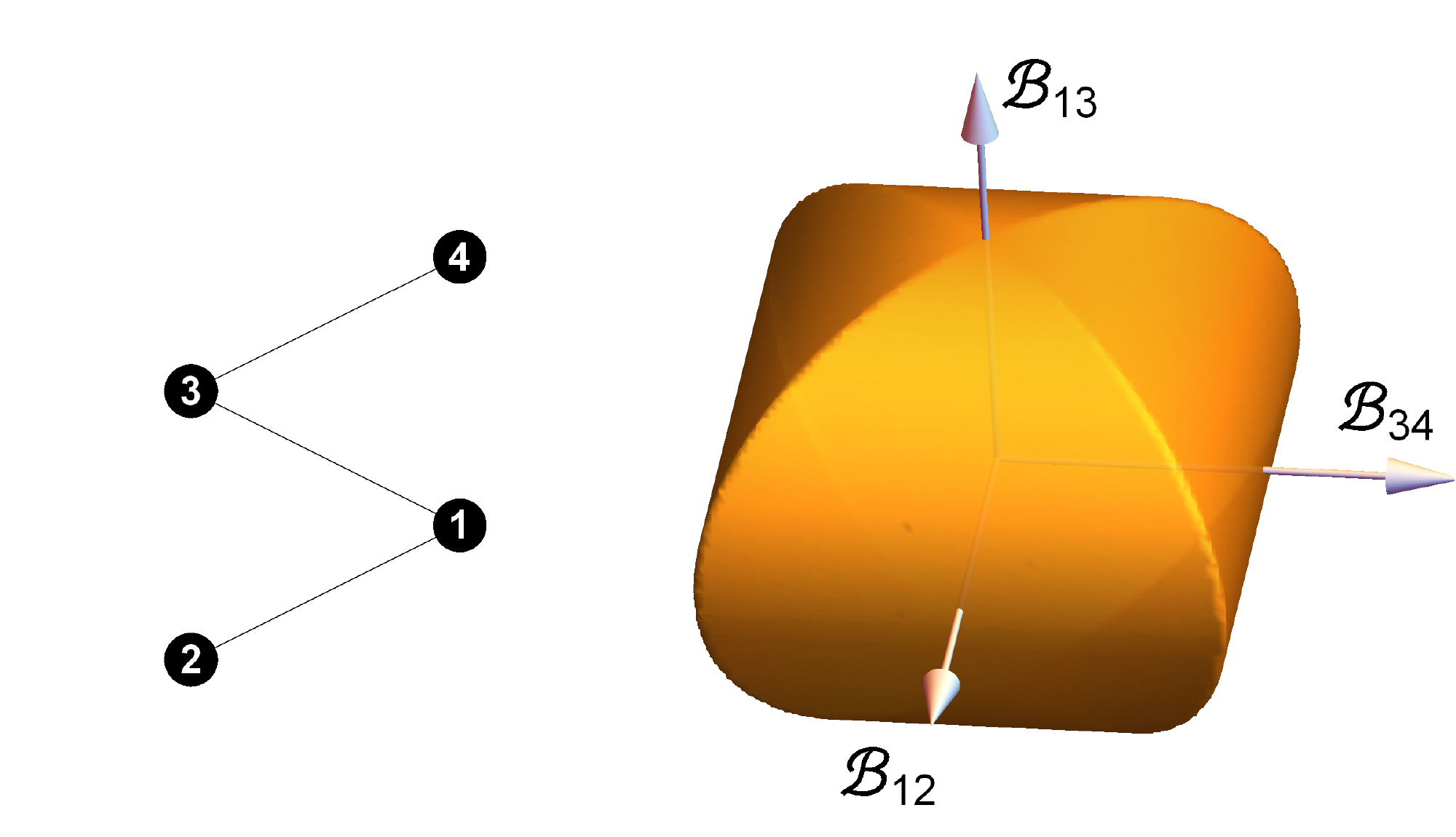}
	\caption{A chain configuration of four observers (left) and the corresponding Steinmetz solid (right).
		This solid fully characterizes allowed values of the three Bell parameters.}
	\label{FIG_APP_CHAIN}
\end{figure}

Here we provide another example where the Steinmetz solid completely characterizes quantum trade-offs between violations of Bell inequalities.
The configuration is illustrated in Fig.~\ref{FIG_APP_CHAIN} and involves three bipartite Bell parameters $\mathcal{B}_{12}, \mathcal{B}_{13}$ and $\mathcal{B}_{34}$, and two elementary monogamy relations:
\begin{eqnarray}
\label{eq:two-elem-mono}
\mathcal{B}^2_{12} + \mathcal{B}^2_{13} \leq 2, \nonumber \\
\mathcal{B}^2_{13} + \mathcal{B}^2_{34} \leq 2.
\end{eqnarray} 
The intersection of the two cylinders above can be split into one of the following parameter regions: 
\begin{enumerate}
	\item[(P1)] $\mathcal{B}_{12}, \mathcal{B}_{34} > 1$, $\mathcal{B}_{13} \leq 1$;
	\item[(P2)] $\mathcal{B}_{12} > 1$, $\mathcal{B}_{34}, \mathcal{B}_{13} \leq 1$;
	\item[(P3)] $\mathcal{B}_{12}, \mathcal{B}_{34} \leq 1$, $\mathcal{B}_{13} > 1$;
	\item[(P4)] $\mathcal{B}_{12}, \mathcal{B}_{13} \leq 1$, $\mathcal{B}_{34} > 1$; 
	\item[(P5)] $\mathcal{B}_{12}, \mathcal{B}_{13}, \mathcal{B}_{34} \leq 1$.
\end{enumerate}
We shall analyze them one by one.
Note that cases (P2) and (P4) are obtained by symmetric interchange of qubits $1 \leftrightarrow 3$ and $2 \leftrightarrow 4$, 
while case (P5) is the region within the local set which is evidently realizable within quantum theory. 

\textbf{Case (P1)} To realize the region defined by (P1), we consider the correlations in the $xz$ plane of the state
\begin{eqnarray*}
| \psi_1 \rangle = \alpha_1 | 0000 \rangle + \alpha_2 | 1111 \rangle + \alpha_3 |1100 \rangle + \alpha_4 |0011 \rangle, 
\end{eqnarray*} 
with $\alpha_i \in \mathbb{R}$ and $\sum_{i=1}^4 \alpha_i^2 = 1$. 
The values of Bell parameters are: 
\begin{eqnarray}
\mathcal{B}^2_{12} &=& 1 + 4\left(\alpha_1 \alpha_3 + \alpha_2 \alpha_4 \right)^2, \nonumber \\
\mathcal{B}^2_{34} &=& 1 + 4 \left( \alpha_2 \alpha_3 + \alpha_1 \alpha_4 \right)^2, \nonumber \\
\mathcal{B}^2_{13} &=& \left(1 - 2 \alpha^2_3 - 2 \alpha^2_4 \right)^2.
\end{eqnarray}
Considering without loss of generality the region where $\mathcal{B}_{12}, \mathcal{B}_{13}, \mathcal{B}_{34} \geq 0$, we choose the following real coefficients 
\begin{eqnarray}
\label{eq:alpha-val}
&&\alpha_1 = \sqrt{\frac{1+\mathcal{B}_{13}}{2}} \cos \phi, \quad \alpha_2 = \sqrt{\frac{1+\mathcal{B}_{13}}{2}} \sin \phi, \nonumber \\
&&\alpha_3 = \sqrt{\frac{1-\mathcal{B}_{13}}{2}} \cos \theta, \quad \alpha_4 = \sqrt{\frac{1-\mathcal{B}_{13}}{2}} \sin \theta,
\end{eqnarray}
for angles $\theta, \phi$ given as
\begin{eqnarray}
\label{eq:theta-phi}
\theta &=& \frac{1}{2} \left[ \arccos \left(\sqrt{\frac{\mathcal{B}^2_{12}-1}{1-\mathcal{B}^2_{13}}} \right) + \arcsin \left(\sqrt{\frac{\mathcal{B}^2_{34} - 1}{1 - \mathcal{B}^2_{13}}} \right) \right], \nonumber \\
\phi &=& \frac{1}{2} \left[ \arcsin \left(\sqrt{\frac{\mathcal{B}^2_{34}-1}{1-\mathcal{B}^2_{13}}} \right) - \arccos \left(\sqrt{\frac{\mathcal{B}^2_{12} - 1}{1 - \mathcal{B}^2_{13}}} \right) \right], \nonumber \\
\end{eqnarray}
which are well-defined for any values $1 < \mathcal{B}_{12}, \mathcal{B}_{34} \leq \sqrt{2}$ and $0 \leq \mathcal{B}_{13} < 1$. 
For $\mathcal{B}_{13} = 1$, Eq.~(\ref{eq:alpha-val}) gives $\alpha_3 = \alpha_4 = 0$. 

\textbf{Case (P2)} The region (P2) is realized by the state
\begin{eqnarray*}
	\rho_2 & = & \proj{\psi_{13}} \otimes \proj{0} \otimes \left(p | 0 \rangle \langle 0 | + \frac{1-p}{2} \mathbb{I} \right) \\
	& + & \proj{\psi_{42}} \otimes \proj{1} \otimes \left(p | 1 \rangle \langle 1 | + \frac{1-p}{2} \mathbb{I} \right),
\end{eqnarray*}
where
\begin{equation*}
\ket{\psi_{jk}} = \alpha_j \ket{00} + \alpha_k \ket{11},
\end{equation*}
with $p, \alpha_i \in \mathbb{R}$, $\sum_{i=1}^{4} \alpha_i^2 = 1$ and $0 \leq p \leq 1$.  
The corresponding Bell parameters calculated in the $xz$ plane are:
\begin{eqnarray}
\mathcal{B}^2_{12} &=& 1 + 4\left(\alpha_1 \alpha_3 + \alpha_2 \alpha_4 \right)^2, \nonumber \\
\mathcal{B}^2_{34} &=& p^2, \nonumber \\
\mathcal{B}^2_{13} &=& \left(1 - 2 \alpha^2_3 - 2 \alpha^2_4 \right)^2.
\end{eqnarray}
Choosing the values of $\alpha_i$ as in Eq.~(\ref{eq:alpha-val}) with the parameters $\theta, \phi$ given in Eq.~(\ref{eq:theta-phi}) allows the realization of any $\mathcal{B}_{12} > 1$ and $\mathcal{B}_{34}, \mathcal{B}_{13} \leq 1$. 

\textbf{Case (P3)} The region (P3) is realized by a state similar to that defined in Sec.~\ref{APP_SEC_STAR}, about the star configuration. 
We compute the values of the Bell parameters from the correlations in the $xz$ plane for the state
\begin{eqnarray}
|\psi_3 \rangle = \alpha |0000 \rangle + \beta |1010 \rangle,
\end{eqnarray}
with $\alpha = \frac{1}{\sqrt{2}} \sqrt{1+\sqrt{2} \sin t}$ and $\beta = \frac{1}{\sqrt{2}} \sqrt{1 - \sqrt{2} \sin t}$, and parameter $0 \leq t \leq \frac{\pi}{4}$. 
This gives
\begin{eqnarray}
\mathcal{B}^2_{13} &=& 2 \cos^2 t, \nonumber \\
\mathcal{B}^2_{12} &=& \mathcal{B}^2_{34} = 2 \sin^2 t.
\end{eqnarray}
To obtain smaller values of $\mathcal{B}_{12}, \mathcal{B}_{34}$, we add noise to the qubits at positions $2$ and $4$ as in Sec.~\ref{APP_SEC_STAR}.


\section{Counting all elementary monogamy relations}
\label{APP_ALG}

This section details a simple brute force algorithm to list all elementary monogamy relations.
We shall focus on $m=3$, in which case the first nontrivial graph has $n = 4$ and $h = 2$, giving rise to monogamy relation (\ref{EQ_2_OVER}).
The set of elementary monogamy relations is denoted as $\mathcal{E}$ and we now add to it the first member given by (\ref{EQ_2_OVER}).
The algorithm then enters a loop as follows:
\begin{itemize}
	
	\item[$\triangleright$] construct all the hypergraphs with $n$ vertices and $h$ hyperedges
	
	\item[$\triangleright$] let index $j$ loop over all these hypergraphs, and let $G_j$ denote the hypergraph corresponding to $j$
	
	\item[$\triangleright$] for each $j$ find the bound on the monogamy relation $M_j$, corresponding to $G_j$, using averaging of the elementary relations in $\mathcal{E}$
	
	\item[$\triangleright$] check if the bound is tight
	
	\begin{itemize}
		\item[$\diamond$] if it is tight, move on to the next $j$
		\item[$\diamond$] if it is not tight, add $M_j$ to $\mathcal{E}$ and remove from $\mathcal{E}$ all monogamy relations that correspond to the subgraphs of $G_j$ and have the same bound as $M_j$
	\end{itemize}
	
	\item[$\triangleright$] loop over $h$ [its maximum number is $n \choose 3$], then loop over $n$
\end{itemize}


\section{Correlations versus local expectation values}
\label{APP_LEGGETT}

For completeness we derive the inequality presented in~\cite{FoundPhys.33.1469,Nature.446.871,NaturePhys.4.681}, which leads to Eq.~\eqref{EQ_MM_BOUND}.
\begin{lemma}
	Consider dichotomic observables $\hat A$, $\hat B$ measured on different sets of particles.
	If $| \langle \hat A \rangle | + | \langle \hat B \rangle | \geq 1$, then
	\begin{align}
	\langle \hat A \otimes \hat B \rangle^2 \geq \left(| \langle \hat A \rangle | + | \langle \hat B \rangle | - 1 \right)^2.
	\label{EQ_APP_B1}
	\end{align}
\end{lemma}
\begin{proof}
	By assumption $\hat A$ and $\hat B$ can be measured simultaneously.
	Let us denote by $A = \pm 1$ and $B = \pm 1$ the measurement outcomes obtained in a single experimental run.
	They satisfy:
	\begin{align}
	-1+ |A + B| = A B = 1 - |A - B|.
	\end{align}
	Averaging over many runs of the experiment gives
	\begin{align}
	-1+ \avg{|A + B|}=\avg{A B} = 1-\avg{|A - B|}.
	\label{EQ_APP_B3}
	\end{align}
	Since the average of modulus is at least the modulus of average, we have
	\begin{align}
	-1+ |\avg{A} + \avg{B} | \leq \avg{A B} \leq 1 - |\avg{A} - \avg{B}|.
	\label{EQ_APP_B6}
	\end{align}
	If signs of $\avg{A}$ and $\avg{B}$ are the same, then $|\avg{A}+\avg{B}| = |\avg{A}|+|\avg{B}|$.
	The left inequality in Eq.~\eqref{EQ_APP_B6} then implies
	\begin{align}
	\avg{A B} \geq -1+|\avg{A}| + |\avg{B}|.
	\end{align}
	By our assumption both sides are nonnegative leading to Eq.~\eqref{EQ_APP_B1}.
	Similarly, if $\avg{A}$ and $\avg{B}$ have opposite signs, then $|\avg{A} - \avg{B}| = |\avg{A}|+|\avg{B}|$, and the right inequality in Eq.~\eqref{EQ_APP_B3} becomes
	\begin{align}
	\avg{A B} \leq 1 - |\avg{A}| - |\avg{B}|.
	\end{align}
	By our assumption, this time both sides are non-positive. 
	Multiplying the inequality with itself results in Eq.~\eqref{EQ_APP_B1}.
	In~\eqref{EQ_APP_B1} we use the operator notation to stress that $\hat A$ and $\hat B$ are measured on different particles.
\end{proof}


\section{Monogamy relations for higher number of observers}
\label{APP_N}

Here we show how to combine $m$-partite Bell monogamy relations, obtained from correlation complementarity, to get tight $(m+1)$-partite relations.

We start from the monogamy relation for three-party inequalities in Eq.~(\ref{EQ_SQUARE}), and use it to derive monogamy between four-party inequalities. 
Recall Fig.~\ref{FIG_MONO_42}, which presents configuration corresponding to Eq.~(\ref{EQ_SQUARE}). 
We construct the network for four-party inequalities by adding a new party, labeled `0', who takes part in the four-party Bell experiment with each of the sets in $\{(123), (234), (341), (412)\}$, as well as with a copy $\{(\tilde{1}\tilde{2}\tilde{3}), (\tilde{2}\tilde{3}\tilde{4}), (\tilde{3}\tilde{4}\tilde{1}), (\tilde{4}\tilde{1}\tilde{2})\}$. 
We derive the following tight monogamy relation from Eq.~(\ref{EQ_SQUARE}), i.e. the "square" graph:
\begin{eqnarray}
&&\mathcal{B}_{0123}^2 + \mathcal{B}_{0234}^2 + \mathcal{B}_{0341}^2 + \mathcal{B}_{0412}^2 + \nonumber \\ 
&&\mathcal{B}_{0\tilde{1}\tilde{2}\tilde{3}}^2 + \mathcal{B}_{0\tilde{2}\tilde{3}\tilde{4}}^2 + \mathcal{B}_{0\tilde{3}\tilde{4}\tilde{1}}^2 + \mathcal{B}_{0\tilde{4}\tilde{1}\tilde{2}}^2\le 8.
\label{EQ_SQUARE_2}
\end{eqnarray} 
To this end, we first recall that Eq.~(\ref{EQ_SQUARE}) was proven by grouping the relevant observables into the following anti-commuting sets
\begin{equation}
\label{TAB_CC_1}
\centering
\begin{tabular}{|c|c|c|c|}
\hline
$X_1X_2Y_3$ & $X_1Y_2X_3$ & $Y_1X_2X_3$ & $Y_1Y_2Y_3$ \\ 
$X_1Y_2X_4$ & $Y_1Y_2Y_4$ & $X_1X_2Y_4$ & $Y_1X_2X_4$ \\ 
$X_1X_3Y_4$ & $Y_1X_3X_4$ & $Y_1Y_3Y_4$ & $X_1Y_3X_4$ \\  
$Y_2Y_3Y_4$ & $X_2X_3Y_4$ & $X_2Y_3X_4$ & $Y_2X_3X_4$ \\ 
$(X_i \leftrightarrow Y_i)$ & $(X_i \leftrightarrow Y_i)$ & $(X_i \leftrightarrow Y_i)$ & $(X_i \leftrightarrow Y_i)$ \\
\hline
\end{tabular}
\end{equation}
where the last line in the table indicates the four observables obtained from the previous lines by an interchange $X_i \leftrightarrow Y_i$ at each site $i$, i.e., $X_1X_2Y_3 \rightarrow Y_1Y_2X_3$, etc, making a total of eight anti-commuting observables in each of the four sets. 
In order to prove Eq.~(\ref{EQ_SQUARE_2}) we form the requisite eight sets of $16$ anti-commuting observables by adjoining $X_0$ to each of the above sets and $Y_0$ to the corresponding sets for the $\tilde{i}$ qubits.  
\begin{equation}
\label{TAB_CC_2}
\centering
\begin{tabular}{|c|c|c|c|}
\hline
$X_0X_1X_2Y_3$ & $X_0X_1Y_2X_3$ & $X_0Y_1X_2X_3$ & $X_0Y_1Y_2Y_3$  \\ 
$X_0X_1Y_2X_4$ & $X_0Y_1Y_2Y_4$ & $X_0X_1X_2Y_4$ & $X_0Y_1X_2X_4$  \\ 
$X_0X_1X_3Y_4$ & $X_0Y_1X_3X_4$ & $X_0Y_1Y_3Y_4$ & $X_0X_1Y_3X_4$  \\  
$X_0Y_2Y_3Y_4$ & $X_0X_2X_3Y_4$ & $X_0X_2Y_3X_4$ & $X_0Y_2X_3X_4$ \\ 
$X_0(X_i \leftrightarrow Y_i)$ & $X_0(X_i \leftrightarrow Y_i)$ & $X_0(X_i \leftrightarrow Y_i)$ & $X_0(X_i \leftrightarrow Y_i)$ \\
$Y_0X_{\tilde{1}}X_{\tilde{2}}Y_{\tilde{3}}$ & $Y_0X_{\tilde{1}}Y_{\tilde{2}}X_{\tilde{3}}$ & $Y_0Y_{\tilde{1}}X_{\tilde{2}}X_{\tilde{3}}$ & $Y_0Y_{\tilde{1}}Y_{\tilde{2}}Y_{\tilde{3}}$  \\ 
$Y_0X_{\tilde{1}}Y_{\tilde{2}}X_{\tilde{4}}$ & $Y_0Y_{\tilde{1}}Y_{\tilde{2}}Y_{\tilde{4}}$ & $Y_0X_{\tilde{1}}X_{\tilde{2}}Y_{\tilde{4}}$ & $Y_0Y_{\tilde{1}}X_{\tilde{2}}X_{\tilde{4}}$ \\ 
$Y_0X_{\tilde{1}}X_{\tilde{3}}Y_{\tilde{4}}$ & $Y_0Y_{\tilde{1}}X_{\tilde{3}}X_{\tilde{4}}$ & $Y_0Y_{\tilde{1}}Y_{\tilde{3}}Y_{\tilde{4}}$ & $Y_0X_{\tilde{1}}Y_{\tilde{3}}X_{\tilde{4}}$ \\  
$Y_0Y_{\tilde{2}}Y_{\tilde{3}}Y_{\tilde{4}}$ & $Y_0X_{\tilde{2}}X_{\tilde{3}}Y_{\tilde{4}}$ & $Y_0X_{\tilde{2}}Y_{\tilde{3}}X_{\tilde{4}}$ & $Y_0Y_{\tilde{2}}X_{\tilde{3}}X_{\tilde{4}}$ \\ 
$Y_0(X_{\tilde{i}} \leftrightarrow Y_{\tilde{i}})$ & $Y_0(X_{\tilde{i}} \leftrightarrow Y_{\tilde{i}})$ & $Y_0(X_{\tilde{i}} \leftrightarrow Y_{\tilde{i}})$ & $Y_0(X_{\tilde{i}} \leftrightarrow Y_{\tilde{i}})$ \\
\hline
\end{tabular}
\end{equation}
In addition to the above four sets, we obtain four more sets of $16$ anti-commuting observables from each of the above sets by interchanging $X_0 \leftrightarrow Y_0$, i.e., $X_0X_1X_2X_3 \rightarrow Y_0X_1X_2X_3$, etc. 

It is readily verified that the observables within each column in Tab.~\ref{TAB_CC_2} anti-commute. 
Moreover, the construction extends so that starting from any network where one has derived a monogamy relation from correlation complementarity for $m$-party inequalities, one can obtain a tight monogamy relation for ($m+1$)-party inequalities in a ``tree'' network, with a central qubit $0$ and with the two ``leaves" corresponding to the original $m$-party network.
Furthermore, the derived inequalities are tight and completely characterize the quantum trade-off relation, 
which is ensured by considering the correlations in the $xy$ plane of the following state \cite{PhysRevLett.106.180402,arXiv:1704.03790}:
\begin{eqnarray}
|\psi \rangle = \frac{1}{\sqrt{2}}\sum_{e} \alpha_e | \underbrace{0 \dots 0}_{e} 1 \dots 1 \rangle + \frac{1}{\sqrt{2}} \ket{1 \dots 1},
\end{eqnarray}  
where $\alpha_e \in \mathbb{R}$ and normalized, $e$ denotes the hyperedge and $|1 \dots 1 \rangle$ denotes the all-$1$ state on the (remaining) qubits in the network.

\bibliographystyle{apsrev4-1}

\begin{thebibliography}{32}%
\makeatletter
\providecommand \@ifxundefined [1]{%
 \@ifx{#1\undefined}
}%
\providecommand \@ifnum [1]{%
 \ifnum #1\expandafter \@firstoftwo
 \else \expandafter \@secondoftwo
 \fi
}%
\providecommand \@ifx [1]{%
 \ifx #1\expandafter \@firstoftwo
 \else \expandafter \@secondoftwo
 \fi
}%
\providecommand \natexlab [1]{#1}%
\providecommand \enquote  [1]{``#1''}%
\providecommand \bibnamefont  [1]{#1}%
\providecommand \bibfnamefont [1]{#1}%
\providecommand \citenamefont [1]{#1}%
\providecommand \href@noop [0]{\@secondoftwo}%
\providecommand \href [0]{\begingroup \@sanitize@url \@href}%
\providecommand \@href[1]{\@@startlink{#1}\@@href}%
\providecommand \@@href[1]{\endgroup#1\@@endlink}%
\providecommand \@sanitize@url [0]{\catcode `\\12\catcode `\$12\catcode
  `\&12\catcode `\#12\catcode `\^12\catcode `\_12\catcode `\%12\relax}%
\providecommand \@@startlink[1]{}%
\providecommand \@@endlink[0]{}%
\providecommand \url  [0]{\begingroup\@sanitize@url \@url }%
\providecommand \@url [1]{\endgroup\@href {#1}{\urlprefix }}%
\providecommand \urlprefix  [0]{URL }%
\providecommand \Eprint [0]{\href }%
\providecommand \doibase [0]{http://dx.doi.org/}%
\providecommand \selectlanguage [0]{\@gobble}%
\providecommand \bibinfo  [0]{\@secondoftwo}%
\providecommand \bibfield  [0]{\@secondoftwo}%
\providecommand \translation [1]{[#1]}%
\providecommand \BibitemOpen [0]{}%
\providecommand \bibitemStop [0]{}%
\providecommand \bibitemNoStop [0]{.\EOS\space}%
\providecommand \EOS [0]{\spacefactor3000\relax}%
\providecommand \BibitemShut  [1]{\csname bibitem#1\endcsname}%
\let\auto@bib@innerbib\@empty
\bibitem [{\citenamefont {Scarani}\ and\ \citenamefont
  {Gisin}(2001{\natexlab{a}})}]{PhysRevLett.87.117901}%
  \BibitemOpen
  \bibfield  {author} {\bibinfo {author} {\bibfnamefont {V.}~\bibnamefont
  {Scarani}}\ and\ \bibinfo {author} {\bibfnamefont {N.}~\bibnamefont
  {Gisin}},\ }\href@noop {} {\bibfield  {journal} {\bibinfo  {journal} {Phys
  Rev. Lett.}\ }\textbf {\bibinfo {volume} {87}},\ \bibinfo {pages} {117901}
  (\bibinfo {year} {2001}{\natexlab{a}})}\BibitemShut {NoStop}%
\bibitem [{\citenamefont {Scarani}\ and\ \citenamefont
  {Gisin}(2001{\natexlab{b}})}]{PhysRevA.65.012311}%
  \BibitemOpen
  \bibfield  {author} {\bibinfo {author} {\bibfnamefont {V.}~\bibnamefont
  {Scarani}}\ and\ \bibinfo {author} {\bibfnamefont {N.}~\bibnamefont
  {Gisin}},\ }\href@noop {} {\bibfield  {journal} {\bibinfo  {journal} {Phys.
  Rev. A}\ }\textbf {\bibinfo {volume} {65}},\ \bibinfo {pages} {012311}
  (\bibinfo {year} {2001}{\natexlab{b}})}\BibitemShut {NoStop}%
\bibitem [{\citenamefont {Barrett}\ \emph {et~al.}(2005)\citenamefont
  {Barrett}, \citenamefont {Linden}, \citenamefont {Massar}, \citenamefont
  {Pironio}, \citenamefont {Popescu},\ and\ \citenamefont
  {Roberts}}]{PhysRevA.71.022101}%
  \BibitemOpen
  \bibfield  {author} {\bibinfo {author} {\bibfnamefont {J.}~\bibnamefont
  {Barrett}}, \bibinfo {author} {\bibfnamefont {N.}~\bibnamefont {Linden}},
  \bibinfo {author} {\bibfnamefont {S.}~\bibnamefont {Massar}}, \bibinfo
  {author} {\bibfnamefont {S.}~\bibnamefont {Pironio}}, \bibinfo {author}
  {\bibfnamefont {S.}~\bibnamefont {Popescu}}, \ and\ \bibinfo {author}
  {\bibfnamefont {D.}~\bibnamefont {Roberts}},\ }\href@noop {} {\bibfield
  {journal} {\bibinfo  {journal} {Phys. Rev. A}\ }\textbf {\bibinfo {volume}
  {71}},\ \bibinfo {pages} {022101} (\bibinfo {year} {2005})}\BibitemShut
  {NoStop}%
\bibitem [{\citenamefont {Masanes}\ \emph {et~al.}(2006)\citenamefont
  {Masanes}, \citenamefont {Acin},\ and\ \citenamefont
  {Gisin}}]{PhysRevA.73.012112}%
  \BibitemOpen
  \bibfield  {author} {\bibinfo {author} {\bibfnamefont {L.}~\bibnamefont
  {Masanes}}, \bibinfo {author} {\bibfnamefont {A.}~\bibnamefont {Acin}}, \
  and\ \bibinfo {author} {\bibfnamefont {N.}~\bibnamefont {Gisin}},\
  }\href@noop {} {\bibfield  {journal} {\bibinfo  {journal} {Phys. Rev. A}\
  }\textbf {\bibinfo {volume} {73}},\ \bibinfo {pages} {012112} (\bibinfo
  {year} {2006})}\BibitemShut {NoStop}%
\bibitem [{\citenamefont {Toner}(2009)}]{ProcRSocA.465.59}%
  \BibitemOpen
  \bibfield  {author} {\bibinfo {author} {\bibfnamefont {B.}~\bibnamefont
  {Toner}},\ }\href@noop {} {\bibfield  {journal} {\bibinfo  {journal} {Proc.
  R. Soc. A}\ }\textbf {\bibinfo {volume} {465}},\ \bibinfo {pages} {59}
  (\bibinfo {year} {2009})}\BibitemShut {NoStop}%
\bibitem [{\citenamefont {Paw{\l}owski}\ and\ \citenamefont
  {Brukner}(2009)}]{PhysRevLett.102.030403}%
  \BibitemOpen
  \bibfield  {author} {\bibinfo {author} {\bibfnamefont {M.}~\bibnamefont
  {Paw{\l}owski}}\ and\ \bibinfo {author} {\bibfnamefont {{\v C}.}~\bibnamefont
  {Brukner}},\ }\href@noop {} {\bibfield  {journal} {\bibinfo  {journal} {Phys
  Rev. Lett.}\ }\textbf {\bibinfo {volume} {102}},\ \bibinfo {pages} {030403}
  (\bibinfo {year} {2009})}\BibitemShut {NoStop}%
\bibitem [{\citenamefont {Paw{\l}owski}(2010)}]{PhysRevA.82.032313}%
  \BibitemOpen
  \bibfield  {author} {\bibinfo {author} {\bibfnamefont {M.}~\bibnamefont
  {Paw{\l}owski}},\ }\href@noop {} {\bibfield  {journal} {\bibinfo  {journal}
  {Phys. Rev. A}\ }\textbf {\bibinfo {volume} {82}},\ \bibinfo {pages} {032313}
  (\bibinfo {year} {2010})}\BibitemShut {NoStop}%
\bibitem [{\citenamefont {Augusiak}\ \emph {et~al.}(2014)\citenamefont
  {Augusiak}, \citenamefont {Demianowicz}, \citenamefont {Paw{\l}owski},
  \citenamefont {Tura},\ and\ \citenamefont {Acin}}]{PhysRevA.90.052323}%
  \BibitemOpen
  \bibfield  {author} {\bibinfo {author} {\bibfnamefont {R.}~\bibnamefont
  {Augusiak}}, \bibinfo {author} {\bibfnamefont {M.}~\bibnamefont
  {Demianowicz}}, \bibinfo {author} {\bibfnamefont {M.}~\bibnamefont
  {Paw{\l}owski}}, \bibinfo {author} {\bibfnamefont {J.}~\bibnamefont {Tura}},
  \ and\ \bibinfo {author} {\bibfnamefont {A.}~\bibnamefont {Acin}},\
  }\href@noop {} {\bibfield  {journal} {\bibinfo  {journal} {Phys. Rev. A}\
  }\textbf {\bibinfo {volume} {90}},\ \bibinfo {pages} {052323} (\bibinfo
  {year} {2014})}\BibitemShut {NoStop}%
\bibitem [{\citenamefont {Ramanathan}\ and\ \citenamefont
  {Horodecki}(2014)}]{PhysRevLett.113.210403}%
  \BibitemOpen
  \bibfield  {author} {\bibinfo {author} {\bibfnamefont {R.}~\bibnamefont
  {Ramanathan}}\ and\ \bibinfo {author} {\bibfnamefont {P.}~\bibnamefont
  {Horodecki}},\ }\href@noop {} {\bibfield  {journal} {\bibinfo  {journal}
  {Phys Rev. Lett.}\ }\textbf {\bibinfo {volume} {113}},\ \bibinfo {pages}
  {210403} (\bibinfo {year} {2014})}\BibitemShut {NoStop}%
\bibitem [{\citenamefont {Toner}\ and\ \citenamefont
  {Verstraete}(2006)}]{arXiv:0611001}%
  \BibitemOpen
  \bibfield  {author} {\bibinfo {author} {\bibfnamefont {B.}~\bibnamefont
  {Toner}}\ and\ \bibinfo {author} {\bibfnamefont {F.}~\bibnamefont
  {Verstraete}},\ }\href@noop {} {\bibfield  {journal} {\bibinfo  {journal}
  {arXiv:quant-ph/0611001}\ } (\bibinfo {year} {2006})}\BibitemShut {NoStop}%
\bibitem [{\citenamefont {Kurzy\'nski}\ \emph {et~al.}(2011)\citenamefont
  {Kurzy\'nski}, \citenamefont {Paterek}, \citenamefont {Ramanathan},
  \citenamefont {Laskowski},\ and\ \citenamefont
  {Kaszlikowski}}]{PhysRevLett.106.180402}%
  \BibitemOpen
  \bibfield  {author} {\bibinfo {author} {\bibfnamefont {P.}~\bibnamefont
  {Kurzy\'nski}}, \bibinfo {author} {\bibfnamefont {T.}~\bibnamefont
  {Paterek}}, \bibinfo {author} {\bibfnamefont {R.}~\bibnamefont {Ramanathan}},
  \bibinfo {author} {\bibfnamefont {W.}~\bibnamefont {Laskowski}}, \ and\
  \bibinfo {author} {\bibfnamefont {D.}~\bibnamefont {Kaszlikowski}},\
  }\href@noop {} {\bibfield  {journal} {\bibinfo  {journal} {Phys Rev. Lett.}\
  }\textbf {\bibinfo {volume} {106}},\ \bibinfo {pages} {180402} (\bibinfo
  {year} {2011})}\BibitemShut {NoStop}%
\bibitem [{\citenamefont {Ramanathan}\ and\ \citenamefont
  {Mironowicz}(2017)}]{arXiv:1704.03790}%
  \BibitemOpen
  \bibfield  {author} {\bibinfo {author} {\bibfnamefont {R.}~\bibnamefont
  {Ramanathan}}\ and\ \bibinfo {author} {\bibfnamefont {P.}~\bibnamefont
  {Mironowicz}},\ }\href@noop {} {\bibfield  {journal} {\bibinfo  {journal}
  {arXiv:1704.03790}\ } (\bibinfo {year} {2017})}\BibitemShut {NoStop}%
\bibitem [{Note1()}]{Note1}%
  \BibitemOpen
  \bibinfo {note} {Individual bipartite Bell inequalities already played an
  important role in this context, see e.g.~\cite {Pawlowski2009,ML}. However,
  intrinsically multi-partite principles are required~\cite
  {PhysRevLett.107.210403}, see e.g.~\cite {NatComms.4.2263}. Note that for
  testing these principles tight monogamy relations involving both bipartite
  and multipartite Bell inequalities are required.}\BibitemShut {Stop}%
\bibitem [{\citenamefont {Ramanathan}\ \emph {et~al.}(2011)\citenamefont
  {Ramanathan}, \citenamefont {Paterek}, \citenamefont {Kay}, \citenamefont
  {Kurzy\'nski},\ and\ \citenamefont {Kaszlikowski}}]{PhysRevLett.107.060405}%
  \BibitemOpen
  \bibfield  {author} {\bibinfo {author} {\bibfnamefont {R.}~\bibnamefont
  {Ramanathan}}, \bibinfo {author} {\bibfnamefont {T.}~\bibnamefont {Paterek}},
  \bibinfo {author} {\bibfnamefont {A.}~\bibnamefont {Kay}}, \bibinfo {author}
  {\bibfnamefont {P.}~\bibnamefont {Kurzy\'nski}}, \ and\ \bibinfo {author}
  {\bibfnamefont {D.}~\bibnamefont {Kaszlikowski}},\ }\href@noop {} {\bibfield
  {journal} {\bibinfo  {journal} {Phys Rev. Lett.}\ }\textbf {\bibinfo {volume}
  {107}},\ \bibinfo {pages} {060405} (\bibinfo {year} {2011})}\BibitemShut
  {NoStop}%
\bibitem [{Note2()}]{Note2}%
  \BibitemOpen
  \bibinfo {note} {Ref.~\cite {PhysRevLett.102.030403} derives the upper bound
  on the average shrinking factor from the principle of no-signaling. Our
  monogamy relations are derived within quantum formalism and put the upper
  bound on the average \protect \emph {squared} shrinking factor. This is a
  stronger result, especially for anisotropic cloning machines~\cite
  {Cerf00}.}\BibitemShut {Stop}%
\bibitem [{\citenamefont {T\'oth}\ and\ \citenamefont
  {G\"uhne}(2005)}]{PhysRevA.72.022340}%
  \BibitemOpen
  \bibfield  {author} {\bibinfo {author} {\bibfnamefont {G.}~\bibnamefont
  {T\'oth}}\ and\ \bibinfo {author} {\bibfnamefont {O.}~\bibnamefont
  {G\"uhne}},\ }\href@noop {} {\bibfield  {journal} {\bibinfo  {journal} {Phys.
  Rev. A}\ }\textbf {\bibinfo {volume} {72}},\ \bibinfo {pages} {022340}
  (\bibinfo {year} {2005})}\BibitemShut {NoStop}%
\bibitem [{\citenamefont {Wehner}\ and\ \citenamefont
  {Winter}(2008)}]{JMathPhys.49.062105}%
  \BibitemOpen
  \bibfield  {author} {\bibinfo {author} {\bibfnamefont {S.}~\bibnamefont
  {Wehner}}\ and\ \bibinfo {author} {\bibfnamefont {A.}~\bibnamefont
  {Winter}},\ }\href@noop {} {\bibfield  {journal} {\bibinfo  {journal} {J.
  Math. Phys.}\ }\textbf {\bibinfo {volume} {49}},\ \bibinfo {pages} {062105}
  (\bibinfo {year} {2008})}\BibitemShut {NoStop}%
\bibitem [{\citenamefont {Wehner}\ and\ \citenamefont
  {Winter}(2010)}]{NewJPhys.12.025009}%
  \BibitemOpen
  \bibfield  {author} {\bibinfo {author} {\bibfnamefont {S.}~\bibnamefont
  {Wehner}}\ and\ \bibinfo {author} {\bibfnamefont {A.}~\bibnamefont
  {Winter}},\ }\href@noop {} {\bibfield  {journal} {\bibinfo  {journal} {New J.
  Phys.}\ }\textbf {\bibinfo {volume} {12}},\ \bibinfo {pages} {025009}
  (\bibinfo {year} {2010})}\BibitemShut {NoStop}%
\bibitem [{\citenamefont {Werner}\ and\ \citenamefont
  {Wolf}(2001)}]{PhysRevA.64.032112}%
  \BibitemOpen
  \bibfield  {author} {\bibinfo {author} {\bibfnamefont {R.~F.}\ \bibnamefont
  {Werner}}\ and\ \bibinfo {author} {\bibfnamefont {M.~M.}\ \bibnamefont
  {Wolf}},\ }\href@noop {} {\bibfield  {journal} {\bibinfo  {journal} {Phys.
  Rev. A}\ }\textbf {\bibinfo {volume} {64}},\ \bibinfo {pages} {032112}
  (\bibinfo {year} {2001})}\BibitemShut {NoStop}%
\bibitem [{\citenamefont {\.Zukowski}\ and\ \citenamefont
  {Brukner}(2002)}]{PhysRevLett.88.210401}%
  \BibitemOpen
  \bibfield  {author} {\bibinfo {author} {\bibfnamefont {M.}~\bibnamefont
  {\.Zukowski}}\ and\ \bibinfo {author} {\bibfnamefont {{\v C}.}~\bibnamefont
  {Brukner}},\ }\href@noop {} {\bibfield  {journal} {\bibinfo  {journal} {Phys
  Rev. Lett.}\ }\textbf {\bibinfo {volume} {88}},\ \bibinfo {pages} {210401}
  (\bibinfo {year} {2002})}\BibitemShut {NoStop}%
\bibitem [{Note3()}]{Note3}%
  \BibitemOpen
  \bibinfo {note} {See e.g. Eq. (5) in Ref.~\cite
  {PhysRevLett.88.210401}.}\BibitemShut {Stop}%
\bibitem [{\citenamefont {Horodecki}\ \emph {et~al.}(1995)\citenamefont
  {Horodecki}, \citenamefont {Horodecki},\ and\ \citenamefont
  {Horodecki}}]{PhysLettA.200.340}%
  \BibitemOpen
  \bibfield  {author} {\bibinfo {author} {\bibfnamefont {R.}~\bibnamefont
  {Horodecki}}, \bibinfo {author} {\bibfnamefont {P.}~\bibnamefont
  {Horodecki}}, \ and\ \bibinfo {author} {\bibfnamefont {M.}~\bibnamefont
  {Horodecki}},\ }\href@noop {} {\bibfield  {journal} {\bibinfo  {journal}
  {Phys. Lett. A}\ }\textbf {\bibinfo {volume} {200}},\ \bibinfo {pages} {340}
  (\bibinfo {year} {1995})}\BibitemShut {NoStop}%
\bibitem [{SM()}]{SM}%
  \BibitemOpen
  \href@noop {} {}\bibinfo {note} {See the Supplemental Material for simplified
  proofs and examples.}\BibitemShut {Stop}%
\bibitem [{\citenamefont {Kafatos}(1989)}]{GHZ}%
  \BibitemOpen
  \bibinfo {editor} {\bibfnamefont {M.}~\bibnamefont {Kafatos}},\ ed.,\
  \href@noop {} {\emph {\bibinfo {title} {Bell's theorem, quantum theory, and
  conceptions of the universe}}}\ (\bibinfo  {publisher} {Kluwer, Dordrecht},\
  \bibinfo {year} {1989})\BibitemShut {NoStop}%
\bibitem [{\citenamefont {Leggett}(2003)}]{FoundPhys.33.1469}%
  \BibitemOpen
  \bibfield  {author} {\bibinfo {author} {\bibfnamefont {A.~J.}\ \bibnamefont
  {Leggett}},\ }\href@noop {} {\bibfield  {journal} {\bibinfo  {journal}
  {Found. Phys.}\ }\textbf {\bibinfo {volume} {33}},\ \bibinfo {pages} {1469}
  (\bibinfo {year} {2003})}\BibitemShut {NoStop}%
\bibitem [{\citenamefont {Gr\"oblacher}\ \emph {et~al.}(2007)\citenamefont
  {Gr\"oblacher}, \citenamefont {Paterek}, \citenamefont {Kaltenbaek},
  \citenamefont {Brukner}, \citenamefont {\.Zukowski}, \citenamefont
  {Aspelmeyer},\ and\ \citenamefont {Zeilinger}}]{Nature.446.871}%
  \BibitemOpen
  \bibfield  {author} {\bibinfo {author} {\bibfnamefont {S.}~\bibnamefont
  {Gr\"oblacher}}, \bibinfo {author} {\bibfnamefont {T.}~\bibnamefont
  {Paterek}}, \bibinfo {author} {\bibfnamefont {R.}~\bibnamefont {Kaltenbaek}},
  \bibinfo {author} {\bibfnamefont {{\v C}.}~\bibnamefont {Brukner}}, \bibinfo
  {author} {\bibfnamefont {M.}~\bibnamefont {\.Zukowski}}, \bibinfo {author}
  {\bibfnamefont {M.}~\bibnamefont {Aspelmeyer}}, \ and\ \bibinfo {author}
  {\bibfnamefont {A.}~\bibnamefont {Zeilinger}},\ }\href@noop {} {\bibfield
  {journal} {\bibinfo  {journal} {Nature}\ }\textbf {\bibinfo {volume} {446}},\
  \bibinfo {pages} {871} (\bibinfo {year} {2007})}\BibitemShut {NoStop}%
\bibitem [{\citenamefont {Branciard}\ \emph {et~al.}(2008)\citenamefont
  {Branciard}, \citenamefont {Brunner}, \citenamefont {Gisin}, \citenamefont
  {Kurtsiefer}, \citenamefont {Lamas-Linares}, \citenamefont {Ling},\ and\
  \citenamefont {Scarani}}]{NaturePhys.4.681}%
  \BibitemOpen
  \bibfield  {author} {\bibinfo {author} {\bibfnamefont {C.}~\bibnamefont
  {Branciard}}, \bibinfo {author} {\bibfnamefont {N.}~\bibnamefont {Brunner}},
  \bibinfo {author} {\bibfnamefont {N.}~\bibnamefont {Gisin}}, \bibinfo
  {author} {\bibfnamefont {C.}~\bibnamefont {Kurtsiefer}}, \bibinfo {author}
  {\bibfnamefont {A.}~\bibnamefont {Lamas-Linares}}, \bibinfo {author}
  {\bibfnamefont {A.}~\bibnamefont {Ling}}, \ and\ \bibinfo {author}
  {\bibfnamefont {V.}~\bibnamefont {Scarani}},\ }\href@noop {} {\bibfield
  {journal} {\bibinfo  {journal} {Nat. Phys.}\ }\textbf {\bibinfo {volume}
  {4}},\ \bibinfo {pages} {681} (\bibinfo {year} {2008})}\BibitemShut {NoStop}%
\bibitem [{\citenamefont {Pawlowski}\ \emph {et~al.}(2009)\citenamefont
  {Pawlowski}, \citenamefont {Paterek}, \citenamefont {Kaszlikowski},
  \citenamefont {Scarani}, \citenamefont {Winter},\ and\ \citenamefont
  {Zukowski}}]{Pawlowski2009}%
  \BibitemOpen
  \bibfield  {author} {\bibinfo {author} {\bibfnamefont {M.}~\bibnamefont
  {Pawlowski}}, \bibinfo {author} {\bibfnamefont {T.}~\bibnamefont {Paterek}},
  \bibinfo {author} {\bibfnamefont {D.}~\bibnamefont {Kaszlikowski}}, \bibinfo
  {author} {\bibfnamefont {V.}~\bibnamefont {Scarani}}, \bibinfo {author}
  {\bibfnamefont {A.}~\bibnamefont {Winter}}, \ and\ \bibinfo {author}
  {\bibfnamefont {M.}~\bibnamefont {Zukowski}},\ }\href@noop {} {\bibfield
  {journal} {\bibinfo  {journal} {Nature}\ }\textbf {\bibinfo {volume} {461}},\
  \bibinfo {pages} {1101} (\bibinfo {year} {2009})}\BibitemShut {NoStop}%
\bibitem [{\citenamefont {Navascues}\ and\ \citenamefont
  {Wunderlich}(2010)}]{ML}%
  \BibitemOpen
  \bibfield  {author} {\bibinfo {author} {\bibfnamefont {M.}~\bibnamefont
  {Navascues}}\ and\ \bibinfo {author} {\bibfnamefont {H.}~\bibnamefont
  {Wunderlich}},\ }\href@noop {} {\bibfield  {journal} {\bibinfo  {journal}
  {Proc. Roy. Soc. A}\ }\textbf {\bibinfo {volume} {466}},\ \bibinfo {pages}
  {881 } (\bibinfo {year} {2010})}\BibitemShut {NoStop}%
\bibitem [{\citenamefont {Gallego}\ \emph {et~al.}(2011)\citenamefont
  {Gallego}, \citenamefont {W\"urflinger}, \citenamefont {Ac\'{\i}n},\ and\
  \citenamefont {Navascu\'es}}]{PhysRevLett.107.210403}%
  \BibitemOpen
  \bibfield  {author} {\bibinfo {author} {\bibfnamefont {R.}~\bibnamefont
  {Gallego}}, \bibinfo {author} {\bibfnamefont {L.~E.}\ \bibnamefont
  {W\"urflinger}}, \bibinfo {author} {\bibfnamefont {A.}~\bibnamefont
  {Ac\'{\i}n}}, \ and\ \bibinfo {author} {\bibfnamefont {M.}~\bibnamefont
  {Navascu\'es}},\ }\href {\doibase 10.1103/PhysRevLett.107.210403} {\bibfield
  {journal} {\bibinfo  {journal} {Phys. Rev. Lett.}\ }\textbf {\bibinfo
  {volume} {107}},\ \bibinfo {pages} {210403} (\bibinfo {year}
  {2011})}\BibitemShut {NoStop}%
\bibitem [{\citenamefont {Fritz}\ \emph {et~al.}(2013)\citenamefont {Fritz},
  \citenamefont {Sainz}, \citenamefont {Augusiak}, \citenamefont {Brask},
  \citenamefont {Chaves}, \citenamefont {Leverrier},\ and\ \citenamefont
  {Ac{\'i}n}}]{NatComms.4.2263}%
  \BibitemOpen
  \bibfield  {author} {\bibinfo {author} {\bibfnamefont {T.}~\bibnamefont
  {Fritz}}, \bibinfo {author} {\bibfnamefont {A.~B.}\ \bibnamefont {Sainz}},
  \bibinfo {author} {\bibfnamefont {R.}~\bibnamefont {Augusiak}}, \bibinfo
  {author} {\bibfnamefont {J.~B.}\ \bibnamefont {Brask}}, \bibinfo {author}
  {\bibfnamefont {R.}~\bibnamefont {Chaves}}, \bibinfo {author} {\bibfnamefont
  {A.}~\bibnamefont {Leverrier}}, \ and\ \bibinfo {author} {\bibfnamefont
  {A.}~\bibnamefont {Ac{\'i}n}},\ }\href {http://dx.doi.org/10.1038/ncomms3263}
  {\bibfield  {journal} {\bibinfo  {journal} {Nat. Comms.}\ }\textbf {\bibinfo
  {volume} {4}},\ \bibinfo {pages} {2263} (\bibinfo {year} {2013})},\ \bibinfo
  {note} {article}\BibitemShut {NoStop}%
\bibitem [{\citenamefont {Cerf}(2000)}]{Cerf00}%
  \BibitemOpen
  \bibfield  {author} {\bibinfo {author} {\bibfnamefont {N.~J.}\ \bibnamefont
  {Cerf}},\ }\href {\doibase 10.1080/09500340008244036} {\bibfield  {journal}
  {\bibinfo  {journal} {Journal of Modern Optics}\ }\textbf {\bibinfo {volume}
  {47}},\ \bibinfo {pages} {187} (\bibinfo {year} {2000})}\BibitemShut
  {NoStop}%
\end{thebibliography}%

%

\end{document}